\documentclass[10pt,journal,compsoc]{IEEEtran}

\usepackage{booktabs}
\usepackage{xcolor}
\usepackage{color}
\usepackage{bbding}
\usepackage{array, caption, threeparttable}
\usepackage{lettrine}
\usepackage{multirow}
\usepackage{graphicx}
\usepackage{subfigure}
\usepackage{amsmath}
\usepackage{amssymb}
\usepackage{amsthm}
\usepackage{url}
\usepackage{boxedminipage}
\usepackage{algorithmic}
\usepackage[linesnumbered,ruled,vlined]{algorithm2e}
\usepackage{multicol}
\usepackage{epstopdf}
\newcommand{\paratitle}[1]{\vspace{0.5ex}\noindent {\textbf{#1}}}
\newtheorem{definition}{\bf Definition}
\newtheorem{theorem}{\bf Theorem}

\graphicspath{{figs/}}
\newcommand{\ie}{\emph{i.e.,}\xspace}
\newcommand{\eg}{\emph{e.g.,}\xspace}

\newcommand{\eat}[1]{}

\begin{document}

\title{VerifyTL: Secure and Verifiable Collaborative Transfer Learning}
\author{Zhuoran~Ma, Jianfeng~Ma, Yinbin~Miao,
        Ximeng~Liu,~\IEEEmembership{Member,~IEEE}, Wei~Zheng, Kim-Kwang Raymond~Choo,~\IEEEmembership{Senior Member,~IEEE,} and Robert H. Deng,~\IEEEmembership{Fellow,~IEEE}

\IEEEcompsocitemizethanks{
\IEEEcompsocthanksitem Z. Ma, J, Ma, Y. Miao, and W. Zheng are with the School of Cyber Engineering, Xidian University, Xi'an 710071, China; Shaanxi Key Laboratory of Network and System Security, Xidian University, Xi'an 710071, China. E-mail: emmazhr@163.com, jfma@mail.xidian.edu.cn, ybmiao@xidian.edu.cn, zhengwei\_1998@aliyun.com
\IEEEcompsocthanksitem X. Liu is with the Key Laboratory of Information Security of Network Systems, College of Mathematics and Computer Science, Fuzhou University, Fuzhou 350108, China. Email: snbnix@gmail.com
\IEEEcompsocthanksitem K.-K. R. Choo is with the Department of Information Systems and Cyber Security, The University of Texas at San Antonio, San Antonio, TX 78249 USA. Email: raymond.choo@fulbrightmail.org
\IEEEcompsocthanksitem R. H. Deng is with the School of Information Systems, Singapore Management University, Singapore 188065. Email: robertdeng@smu.edu.sg}
}

\IEEEtitleabstractindextext{
\begin{abstract}
Getting access to labelled datasets in certain sensitive application domains can be challenging. Hence, one often resorts to transfer learning to transfer knowledge learned from a source domain with sufficient labelled data to a target domain with limited labelled data. However, most existing transfer learning techniques only focus on one-way transfer which brings no benefit to the source domain. In addition, there is the risk of a covert adversary corrupting a number of domains, which can consequently result in inaccurate prediction or privacy leakage. In this paper we construct a secure and \textbf{Verif}iable collaborative \textbf{T}ransfer \textbf{L}earning scheme, VerifyTL, to support two-way transfer learning over potentially untrusted datasets by improving knowledge transfer from a target domain to a source domain. Further, we equip VerifyTL with a cross transfer unit and a weave transfer unit employing SPDZ computation to provide privacy guarantee and verification in the two-domain setting and the multi-domain setting, respectively. Thus, VerifyTL is secure against covert adversary that can compromise up to $n-1$ out of $n$ data domains. We analyze the security of VerifyTL and evaluate its performance over two real-world datasets. Experimental results show that VerifyTL achieves significant performance gains over existing secure learning schemes.
\end{abstract}

}

\maketitle

\IEEEraisesectionheading{\section{Introduction}}
\lettrine[lines=2]{W}{ith} the increasing deployment of Internet of Things (IoT) devices and digitalization of our society, the amount of digital data generated and collected will also increase significantly. This also contributes to renewed interest in Artificial Intelligence (AI), such as deep learning techniques. For example, Convolutional Neural Network (CNN)~\cite{lawrence1997face} has been widely used to facilitate image processing, facial recognition and fingerprint identification. The construction of data-driven CNN model typically requires intensive data resources for analysis and recognition. However, sharing data across systems may not be easy in practice, for example due to security and privacy concerns~\cite{miao2019privacy,miao2019secure,gdpr}. In additional, labeled datasets that can be used in AI model training may also be limited in certain sensitive application domains.

Transfer learning~\cite{pan2009survey} can potentially be used to overcome such a limitation, by transferring knowledge learned on one dataset/application domain to another dataset/application domain. However, as shown in Fig.~\ref{intro}, existing transfer learning mechanisms have a number of limitations, such as the following:
a) In conventional transfer learning, a source domain contributes knowledge to a target domain with no payoff. However, such ``selfless" behavior may not be realistic, as data collection, curation, labeling, etc, come at a cost. In other words, there may be a shortage of source datasets that can be used for transfer learning.
b) The transferred knowledge may be vulnerable to inference attacks (e.g., membership attacks~\cite{shokri2017membership} and reconstruction attacks~\cite{oh2019exploring}), which can result in disclosure of the training data~\cite{wang2019beyond}. Existing secure transfer learning schemes over training data~\cite{miao2018enabling,miao2018hybrid}, however, incur significant computation and communication overheads.
c) There may exist a covert adversary~\cite{aumann2007security} in various data domains, which attempts to corrupt the changing set of data domains. The adversary can attempt to compromise the transfer learning with negative transfer~\cite{gui2018negative,gu2017badnets}, for example by executing a malicious computation, and changing the computation results for transfer learning. Thus, a covert adversary can launch a malicious learning with dishonest majority~\cite{bhuiyan2016collusion} by maliciously tuning transfer learning, resulting in transfer learning behaving badly on specific attacker-chosen inputs. Existing secure machine learning schemes are not generally designed to the setting of dishonest majority.

\begin{figure}
  \centering
  \includegraphics[width=2.8in]{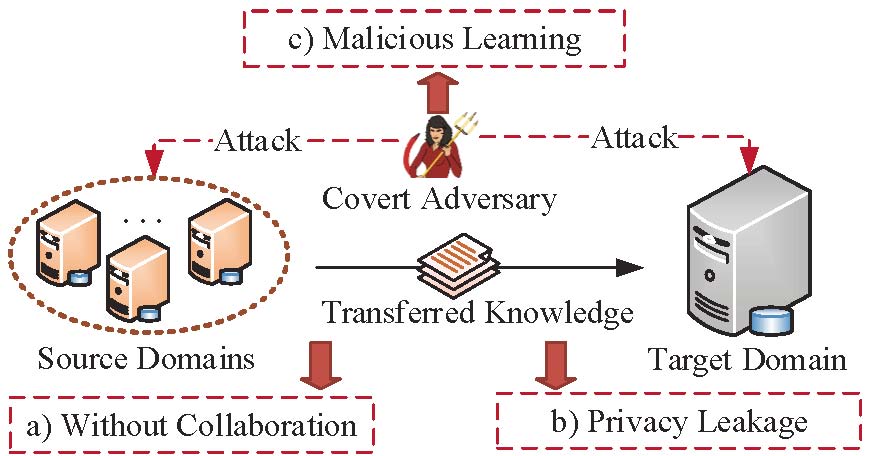}\\
  \caption{Transfer learning framework.}\label{intro}
\end{figure}

The lack of two-way transfer learning and dishonest majority mitigation schemes motivate us to design a secure collaborative training for transfer learning over multiple data-poor domains for implementing covert security in this paper. Specifically, in this paper, we present a {s}ecure and {ver}ifiable {c}ollaborative {transfer} learning against covert adversaries, hereafter referred to as {VerifyTL}. A summary of our contributions is as follows:

\begin{itemize}
  \item \textit{Two-way transfer learning}. We present a collaborative transfer learning framework over multiple data-poor domains, which achieves two-way transfer learning by eliminating the difference between a source domain and a target domain. To tune CNN models on data domains, we provide flexible transfer for collaborative transfer units by setting different knowledge contribution degrees in the respective data domains.

\item \textit{Lightweight privacy preservation}. We propose a lightweight privacy preservation scheme by adopting a locally pre-trained model to extract representations of a data domain. The extracted representations are transferred among data domains to minimize the costs associated with the secure outsourcing of sensitive training data over multiple domains.

  \item \textit{Verifiable learning}. We design a SPDZ-based transfer unit to improve security and support verification. Specifically, we deploy a cross unit and a weave unit, which are two kind of collaborative transfer units, in two-domain setting and multi-domain setting, respectively. The SPDZ-based transfer unit not only securely transfers representations among data domains, but also verifies the correctness of final transferred representations with Message Authentication Code (MAC) to prevent malicious behaviors.

   \item \textit{Covert security}. VerifyTL is a decentralized learning system with covert security, in which a covert adversary can corrupt $n-1$ out of $n$ data domains. Each data domain only trusts itself to prevent the corruption of dishonest majority.

\end{itemize}

In the next two sections, we will review the related literature and introduce relevant background materials. In Section \ref{section:Problem Formulation}, we will introduce the system model of VerifyTL, the threat model of covert security, and the design goals. In Section \ref{section:Proposed VerifyTL}, we present the proposed VerifyTL, prior to evaluating its security and performance in Sections \ref{section:Security Analysis} and \ref{section:Performance Evaluation}. In the last section, we conclude this paper.

\section{Related Work}
Transfer learning has been shown to have potential in the settings where there is a lack of labeled data in one application domain, but knowledge learned from other application domain(s) can be transferred~\cite{pan2009survey,torrey2010transfer}. Earlier approaches mainly focus on transferring the training data from one or more source domains to another~\cite{raina2007self,dai2007boosting}. However, such approaches either incur significant communication costs during the transmission of large amounts of data from the source domain or do not support heterogeneous transfer among different feature distributions. The existing schemes such as those presented in~\cite{yao2010boosting,wang2015online} used a TrAdaBoost approach, which reuses training data of source domains for implementing knowledge transfer. However, TrAdaBoost requires access to training data on both source and target domains. Consequently, the target domain can learn the training dataset of the source domain(s).

Oquab \textit{et al.}~\cite{oquab2014learning,long2018transferable} proposed a CNN-based transfer learning scheme that transfers image representations learned with CNNs on large-scale annotated datasets to other tasks with limited training data. Shin \textit{et al.}~\cite{shin2016deep} designed a transfer learning method that transfers fine-tuning CNN models pre-trained from natural image dataset to medical image tasks for image diagnosis.
Kendall \textit{et al.}~\cite{kendall2018multi} presented a principal approach to multi-task deep learning, which weighs multiple loss functions by considering the homoscedastic uncertainty of each source task. However, these schemes only support one-way transfer learning (\ie knowledge is transferred from a source domain to a task domain). In the event where neither a source domain nor a target domain can collect sufficient labeled data, knowledge is required to be transferred between two data-poor domains. Hence, two-way transfer learning methods such as those presented in~\cite{misra2016cross,hu2018conet} use a cross-stitch network with CNN models for multi-task learning. These methods enable dual knowledge transfers across domains by utilizing cross connections from one task to another and vice versa. However, these two-way transfer learning methods are confined to the two-domain setting, but not the multi-domain setting. In addition, the transferred knowledge may be revealed, for example by successfully carrying out an inference attack over the data representations to reconstruct the training data~\cite{shokri2017membership,wang2019beyond}. In other words, there is a risk of information leakage.

Hence, in recent times, there have been attempts to design privacy-perserving transfer learning approaches, for example by utilizing homomorphic encryption~\cite{ma2019privacy,chen2019fedhealth,salem2019utilizing,liu2018secure} and Multi-Party Computation (MPC)~\cite{bonawitz2017practical,xu2019verifynet}. In~\cite{ma2019privacy}, for example, the training data on each domain are encrypted using homomorphic encryption, prior to been utilized for machine learning. However, the communication overhead increases with the training data size, and a large number of secure computations are required for training over encrypted data (\ie significant computation cost). Other approaches, such as those of~\cite{salem2019utilizing,ma2020pmkt,liu2018secure,chen2019fedhealth}, transfer the model parameters instead of the training data, in order to minimize communication overhead and guarantee data privacy through the model. A homomorphic encryption-based transfer learning scheme is proposed in ~\cite{salem2019utilizing}, which encrypts features extracted from user data and outsourced to honest-but-curious servers. In~\cite{liu2018secure,chen2019fedhealth}, the source domains first pre-train individual models over training datasets, and then encrypt their model parameters for implementing secure outsourcing. These schemes are secure against passive adversaries under the assumption of honest-but-curious entities, where these entities are required to follow the predefined protocols.
Ma \textit{et al.}~\cite{ma2020pmkt} designed a privacy-preserving multi-party knowledge transfer scheme based on decision trees, which preserves the privacy of transferred knowledge from multiple source domains to a target domain. However, in real-world applications, it is not realistic to blindly trust that all entities will strictly follow the protocols. For example, there is the risk of a dishonest majority of source domains who are unwilling to share his/her knowledge and deviates from the predefined protocols during the learning process. In such a scenario, the honest-but-curious assumption will no longer hold; thus, such approaches are potentially vulnerable to the setting of dishonest majority~\cite{damgaard2013practical,liu2018secure,chen2019fedhealth}.

To remove the unrealistic honest-but-curious assumption, the concept of covert security is introduced, which can prevent dishonest majority from deviating the predefined protocols~\cite{aumann2007security,hazay2008efficient}. For example, Zheng \textit{et al.}~\cite{zheng2019helen} presented coopetitive learning (\ie cooperative and competitive) with SPDZ~\cite{araki2018generalizing} to implement covert security, where SPDZ is a practical MPC protocol extended to the dishonest setting. Sharma \textit{et al.}~\cite{sharma2019secure} designed SPDZ-based transfer learning to implement one-way transfer learning with unreliable entities for covert security. However, the scheme only transfers knowledge from the source domain to the target domain in one way, which is clearly unsuitable in data-poor domains as all domains lack learned knowledge.
A comparative summary is presented in Table~\ref{campare}.

\begin{table}[!ht]
 \centering
 \scriptsize
 \caption{Transfer learning approaches: A comparative summary}
 \label{campare}
 \tabcolsep 4pt
 \begin{threeparttable}[b]
 \begin{tabular*}{3.5in}{c||c|c|c|c|c|c}
  \toprule
 \hline
  Approach &\textsf{Fun}$_1$ &\textsf{Fun}$_2$&\textsf{Fun}$_3$&\textsf{Fun}$_4$ &  \textsf{Fun}$_5$&\textsf{Fun}$_6$ \\ \hline
  \cite{yao2010boosting}&SVM&One-way&\Checkmark&\XSolidBrush &\XSolidBrush& \XSolidBrush\\  \hline
  \cite{long2018transferable}&CNN&One-way&\XSolidBrush&\Checkmark  &\XSolidBrush& \XSolidBrush \\ \hline
  \cite{misra2016cross}&CNN&Two-way&\XSolidBrush&\XSolidBrush &\XSolidBrush& \XSolidBrush \\ \hline
  \cite{salem2019utilizing}&CNN&One-way&\Checkmark&\Checkmark &Semi-honest& \XSolidBrush \\ \hline
  \cite{liu2018secure}&Deep learning&One-way&\Checkmark&\Checkmark&Semi-honest& \XSolidBrush \\ \hline
  \cite{zheng2019helen}&Linear model&\XSolidBrush&\Checkmark&\Checkmark&Covert& \Checkmark \\ \hline
  \cite{sharma2019secure}&CNN&One-way&\Checkmark&\Checkmark&Malicious& \XSolidBrush \\ \hline
  {VerifyTL} &CNN&Cross&\Checkmark&\Checkmark&Covert&\Checkmark \\ \hline
  \bottomrule

 \end{tabular*}
  \begin{tablenotes}
\item \textbf{Notes}.  \textsf{Fun}$_1$: Machine learning model; \textsf{Fun}$_2$: One-way or Two-way transfer learning; \textsf{Fun}$_3$: Whether supporting multiple parties or not; \textsf{Fun}$_4$: Whether achieving lightweight transmission or not; \textsf{Fun}$_5$: Semi-honest or Malicious or Covert security model; \textsf{Fun}$_6$: Whether supporting verification or not.
\end{tablenotes}
 \end{threeparttable}
\end{table}

\section{Preliminaries}
We will now briefly describe CNN, transfer learning, and the SPDZ protocol~\cite{hazay2010note} in Sections \ref{subsection:Convolutional Neural Network} to \ref{subsection:SPDZ Protocol}.

\subsection{Convolutional Neural Network (CNN)} \label{subsection:Convolutional Neural Network}
We adopt a CNN as the base model $\mathcal{N}et$ that consists of convolution layers, pooling layers and fully connected layers. Let $\mathcal{X}^0$ and $\mathcal{X}^L=y$  respectively be the input and the desired output, where $L$ is the number of layers and $\mathcal{X}^l$ is the activation map of layer $l\in\{1,...,L\}$.
\begin{itemize}

  \item {Convolution layer $Conv$}: A $Conv$ layer inputs feature maps $\mathcal{X}^{l-1}$ and adopts the sliding convolutional kernels $ker$ for feature extraction. Given an input $\mathcal{X}^{l-1}\in \mathbb{R}^{h_l\times w_l \times c_l}$ in 3rd-order tensors (\ie an array of matrixes) with the height $h_l$, width $w_l$ and channels $c_l$. A $ker$ maps $\mathcal{X}^{l-1}$ to a weighted-sum $\mathcal{X}^{l}$ as defined in $\mathcal{X}^{l}=f(W^l\mathcal{X}^{l-1}),$
where $W^l$ is a weight set of the $l$-th $Conv$ layer.
  \item {Pooling layer $Pool$}: A $Pool$ layer reduces the data dimensions and trainable parameters in the network, and the neurons in this layer are the outputs of a cluster of neurons at the previous layer.
 \item {Activation function $\text{ReLU}$}: The activation function is denoted as a Rectified Linear Unit $\text{ReLU}(x)=\text{max}(0,x)$, which significantly accelerates the training phase and prevents overfitting.
 \item {Full connection layer $Full$}: A $Full$ layer fully connects all its neurons to each neuron at another layer.
 Given an input $\mathcal{X}^{l-1}$, the $l$-th full connection layer outputs $\mathcal{X}^{l}=\text{ReLU}(W^l\mathcal{X}^{l-1}+b^l),$
where $b^l$ is a bias term.

\end{itemize}

\subsection{Transfer Learning}

Transfer learning is a machine learning technique that focuses on acquiring knowledge over data domains (\ie source domains) and repurposing it on a related data domain (\ie target domain). Generally, transfer learning comprises the following three steps:
\begin{itemize}
  \item \textit{Extract knowledge}. A source model is first pre-trained over a source domain, prior to extracting knowledge from training data and repurposing for the target domain.
  \item \textit{Transfer knowledge}. A source domain transfers extracted knowledge to a target domain for the construction of a target model.
  \item \textit{Tune target model}. The target model needs to be refined over the transferred knowledge and the target domain's training data for model optimization.
\end{itemize}

\subsection{SPDZ Protocol} \label{subsection:SPDZ Protocol}
SPDZ protocol, a state-of-the-art MPC protocol, is design to mitigate covert adversaries with secret sharing-based MACs, and tolerate corruption of the majority of parties. More specifically, the SPDZ protocol is divided into online and offline phases. The offline phase performs all computationally expensive public-key operations to create and pre-share triples. The online phase only involves lightweight primitives. The advantages of SPDZ are summarized as follows.

\begin{itemize}
  \item \textit{Privacy}. Given a plaintext $x$, it is converted into $n$ additive shares $x^{(i)}\in \mathbb{Z}_{2^\kappa}$, where $x\equiv\sum x^{(i)} \mod {2^\kappa}$, $\kappa\in\{8,16,32,64,128,...\}$ is the security parameter. The privacy of these shares $x^{(i)}$ is guaranteed by the additive secret sharing.
  \item \textit{Verifiability}. The correctness of all inputs and outputs in SPDZ is verified by the MAC-check mechanism~\cite{damgaard2013practical} with additive secret shares of MACs over the ring of integers $\mathbb{Z}_{2^\kappa}$. For $n$ parties, each party $\mathcal{D}_i$ owns an additive share $\alpha_i\in\mathbb{Z}_{2^\kappa}$ of the fixed MAC key $\alpha$, \ie $\alpha=\alpha_1+ \alpha_2+...+\alpha_n$.
Here, we define $x\in \mathbb{Z}_{2^\kappa}$ is $[\cdot]$-shared when a party holds a tuple $(x^{(i)},\gamma(x)^{(i)})$, where $\gamma(x)^{(i)}$ is an additive share of the corresponding MAC value $\gamma(x)$ as
\begin{equation}
\begin{aligned}\label{maccheck}
\gamma(x)&=\sum \gamma(x)^{(i)} \mod {2^\kappa}
&=\alpha x.
\end{aligned}
\end{equation}

\end{itemize}

\section{System and Threat Models, and Design Goals} \label{section:Problem Formulation}
In this section, we will describe the system and threat models, and the design goals.

\subsection{System Model}
As depicted in Fig.~\ref{systemmodel}, the system model consists of $n$ data-poor domains (\ie $\mathcal{D}_i$, $i\in [1,n]$).
Due to the limitation of training data, it is challenging to construct a high-performance model over a single data domain. Therefore, it is necessary to construct collaborative transfer learning by exchanging extracted knowledge with each other.
Note that a data domain $\mathcal{D}_i$ is not only a source domain for transferring individual knowledge to other domains, but it is also a target domain to leverage exchanged knowledge from the other domains. During the transfer learning phase, there may exist dishonest domains. A dishonest majority of data domains may undermine the learning phase by behaving maliciously to learn the privacy of other domains. Since the confidentiality of exchanged knowledge and learning computation needs to be guaranteed, it is important to construct secure and verifiable collaborative transfer learning.

The entities in our system model perform the following steps. First, each domain $\mathcal{D}_i$ pre-trains a model $\mathcal{N}et_i$. To train the $l$-th layer, $\mathcal{N}et_i$ extracts knowledge over the original data. The extracted knowledge is denoted as data representations, where it is then split into $n$ shares and broadcasted to other $n-1$ domains (Step~\textcircled{\small{1}}). During the collaborative transfer learning phase, a transfer unit realizes secure collaborative transfer learning among $n$ domains, and verifies the computation process to prevent $n-1$ corrupted domains (Step~\textcircled{\small{2}}). Then, the transferred representations are returned to other data domains for tuning each local model (Step~\textcircled{\small{3}}). Steps~\textcircled{\small{1}}-\textcircled{\small{3}} iterate over multiple data domains until a local CNN model is constructed (Step~\textcircled{\small{4}}).

\subsection{Threat Model}\label{sec:threatmodel}
In our threat model, a set of mutually distrustful data domains $\mathcal{D}$ needs to securely implement an agreed computation protocol over their secret inputs. The protocol should be securely executed to implement covert security. It indicates that a changing number of corrupted domains cannot learn additional information, or even lead to incorrect results. To faithfully simulate the adversarial setting, the threat model is defined in the presence of covert adversaries~\cite{aumann2010security}.

\textbf{Covert security model}. Covert adversaries may arbitrarily deviate from the agreed protocol, and at the same time attempting to avoid being caught.  Generally, covert adversaries lie between the following two adversary models, namely: the  \textit{semi-honest} model and the \textit{malicious} model.

\begin{itemize}
  \item \textit{Semi-honest model}. In the passive adversarial setting, a {semi-honest} adversary faithfully follows predefined protocols but may attempt to infer sensitive information from the other domains. Such an adversary cannot collude with other semi-honest domains.

  \item \textit{Malicious model}. In the active adversarial setting, a {malicious} adversary can lead corrupted data domains by arbitrarily deviating from the pre-defined protocol. A group of corrupted data domains can be an arbitrary proportion of $\mathcal{D}$, even $n-1$.
\end{itemize}

\textbf{Note}. The malicious model is stronger than the semi-honest model. However, a semi-honest model is not a special case of the malicious model~\cite{hazay2010note}. The reason is that an adversary in the ideal malicious model can tamper with inputs and outputs, but an adversary in the ideal semi-honest model is not capable of doing so.

\begin{figure}
  \centering
  \includegraphics[width=3.5in]{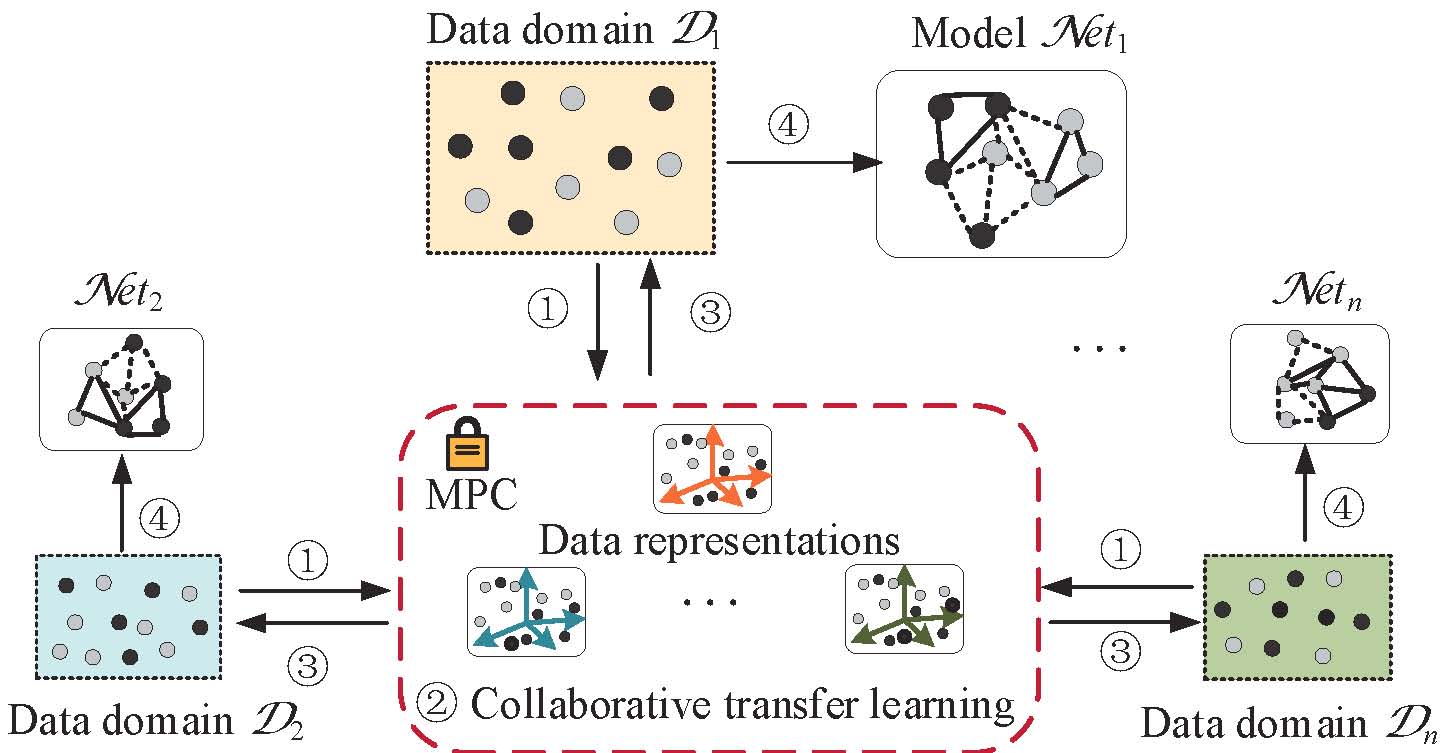}\\
  \caption{System model.}\label{systemmodel}
\end{figure}
\subsection{Design Goals}
To achieve secure and verifiable collaborative transfer learning over multiple data-poor domains, VerifyTL is designed to realize the following goals:
\begin{itemize}

  \item \textit{Covert Security}. To achieve privacy preservation, any data domains in VerifyTL should not learn any other information (including the private inputs and intermediate operations) from the execution process, even in the presence of $n-1$ corrupted domains.
  \item \textit{Verifiability}. Considering that data domains are mutually-untrusted, VerifyTL should verify the correctness of execution process.
  \item \textit{Effectiveness}. Each data domain should play the roles of both source domain and target domain in the collaborative transfer phase, which aims to contribute its knowledge to others and tune its CNN model over the transferred knowledge.
\end{itemize}

\section{Proposed VerifyTL} \label{section:Proposed VerifyTL}
Here, we first summarize a technical overview of VerifyTL, and then design the secure cross unit to implement VerifyTL in the two-domain setting, finally propose the secure weave unit to extend VerifyTL to the multi-domain setting.

\subsection{Technical Overview}
The main motivation behind collaborative transfer learning is that a source domain has no profits during transfer learning.
Thus, we utilize the collaborative transfer learning to realize two-way transfer, which can contribute multi-domain transferred knowledge to tune CNNs.
The core of VerifyTL relies on the following observations:
\begin{itemize}
  \item Data representations on CNNs contain extracted knowledge of original datasets.
  \item According to the correlation extent between two domains, each data domain can set different contribution degree to tune its CNN model over transferred representations from other domains.
  \item A covert adversary can corrupt any data domains, which leads to privacy leakage and malicious computation over dishonest majority.
\end{itemize}

\begin{figure}
  \centering
  \includegraphics[width=3.5in]{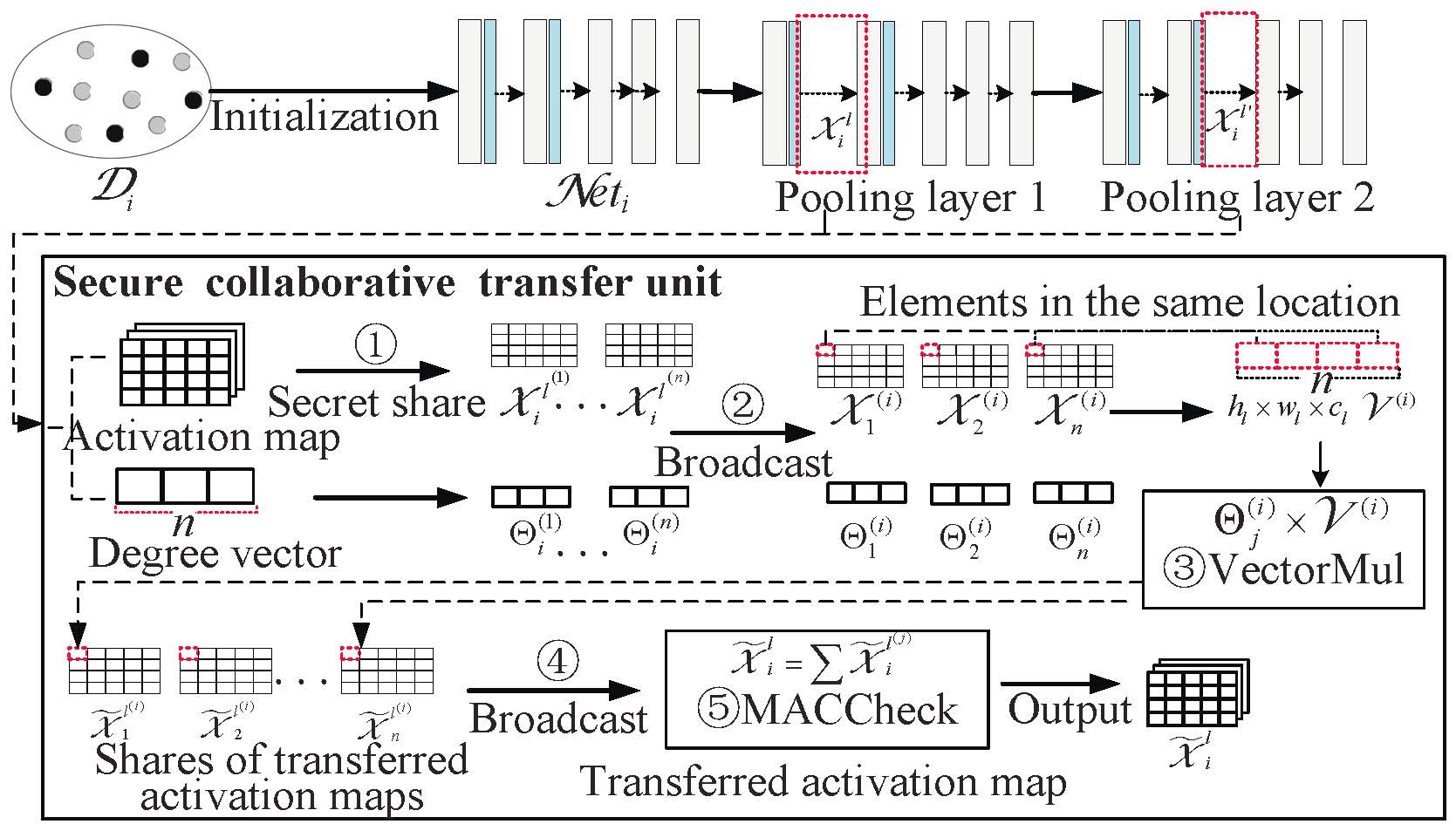}\\
  \caption{Overview of VerifyTL system.}\label{overview}
\end{figure}

In this section, we design the underlying countermeasures based on above observations to achieve VerifyTL:
\begin{itemize}
  \item Data representations contain sensitive information of training data, and thus it is necessary to provide security guarantees.
  \item In each data domain $\mathcal{D}_i$, a degree vector $\Theta_i$ is used to control different contribution degrees of other domains to implement flexible transfer learning.
  \item Collaborative transfer units, \ie two-domain unit (cross unit) and multi-domain unit (weave unit) are proposed to provide secure and verifiable computation for the collaborative transfer learning against dishonest behaviours.
\end{itemize}

Fig.~\ref{overview} illustrates the main process of VerifyTL. Assume that there are $n$ data domains, each of which owns a local training dataset. All data domains agree on the same CNN architecture in advance. Each data domain initializes a CNN model on training data, a pooling layer in a CNN model extracts activation maps as the the inputs of secure collaborative transfer learning. The SPDZ-based cross and weave units are responsible for maintaining secure and verifiable collaborative transfer.
To address the issue of privacy leakage, we propose $\Pi_{\text{CrossUnit}}^{\text{SPDZ}}$ and $\Pi_{\text{WeaveUnit}}^{\text{SPDZ}}$ to protect the transferred data representations and intermediate computation results under the settings of two-domain and multi-domain, respectively (Step $\textcircled{1}$-$\textcircled{4}$).
To avoid the threat of malicious behaviors, we design a \textsf{MACCheck} mechanism to verify the correctness of inputs and computation results (Step $\textcircled{5}$). The notation definitions are listed in Table~\ref{notation}.
The details are described in the following sections.

\begin{table}[!ht]
\centering \caption{Notation descriptions}
\label{notation}
\tabcolsep 0.4pt
\begin{threeparttable}[b] %
\begin{tabular*}{3.5in}{l||l}
\toprule
\hline
Notations & Descriptions \\ \hline
$[x]$-shared&Each data domain holds a tuple $(x^{(i)},\gamma(x)^{(i)})$\\ \hline
$\mathcal{D}$ &The data domain set with the size of $n$\\ \hline
$\mathcal{X}^l$&Activation map of $l$-th layer  \\ \hline
$N(0,1)$&The distribution of zero mean and unit standard deviation \\ \hline
$\widetilde{\mathcal{X}}^l$&Transferred activation map of $l$-th layer\\ \hline
$p$&Precision  \\ \hline
$L$ &Number of layers \\ \hline
$h_l$, $w_l$, $c_l$& the height, width and channels of $\mathcal{X}^l$\\ \hline
$\Theta$, $\Theta_i$ & Degree matrix and degree vector of $\mathcal{D}_i$\\ \hline
$\mathcal{V}$ &Vector of elements at a certain location in all maps\\ \hline
$\mathcal{N}et_i$ &CNN network of the $i$-th data domain \\ \hline
$W_i$&Model parameters $\{W_i^l\}_l^L$ \\ \hline
$\mathcal{L}$&Loss function \\ \hline
$\alpha$&Global MAC key\\ \hline
$\gamma(x)=\alpha x$&MAC value of $x$\\ \hline
$\otimes $& SPDZ multiplication computation over integers\\ \hline
\bottomrule
\end{tabular*}
\end{threeparttable} %
\end{table}

\subsection{Construction of Secure Cross Unit}
Here, we design the cross unit to train networks over data representations transferred between two data domains (\ie $\mathcal{D}_1$ and $\mathcal{D}_2$). A cross unit is employed to implement collaborative transfer learning over activation maps after a pooling layer.
Fig.~\ref{crossconnection} plots the specified process of cross unit between two CNN models $\mathcal{N}et_1$ and $\mathcal{N}et_2$.

\paratitle{Representation Extraction} (Step $\textcircled{1}$):
Assume that the architecture of a CNN model contains two pooling layers. Each data domain pre-trains individual CNN model with the same  architecture over its training data.
At each layer of the network, activation maps~\cite{yang2019towards,zhu2019learning} are proposed as data representations of training data. The overwhelming majority of modern CNN architectures achieve activation maps through a ReLU, which imposes a hard constraint on the intrinsic structure of the maps. The activation map at the $l$-th layer is denoted as $\mathcal{X}^l\leftarrow \mathbb{R}^{h_l\times w_l\times c_l}$.
Then, each data domain implements collaborative transfer learning. After a pooling layer, a cross unit $\Pi_{\text{CrossUnit}}^{\text{SPDZ}}$ is adopted to transfer activation maps between two data domains.

\paratitle{Quantization} (Step $\textcircled{2}$):
In a CNN network, activation maps are normalized with batch normalization~\cite{ioffe2015batch}, and the distribution of each activation map is $N(0,1)$. However, activation maps cannot be directly encoded and operated in SPDZ libraries and thus require pre-process.
We adopt an approximation method to convert floating-point numbers to fixed-point numbers with a precision $p$, where $p$ is the number of bits of approximation precision and the upper bound of the approximation $[1-2^{p},-1+2^{p}]$. For example, given a message $m_1$ and $m_2$, the encoded numbers are defined as $m'_1=Q(m_1,p)=\lfloor m_1 2^p\rceil$ and $m'_2=Q(m_2,p)= \lfloor m_2 2^p\rceil$.
Specifically, the result of multiplication operation in SPDZ can change the precision of $m'_1\times m'_2$ to $2^{2p}$ while the result of an addition operation $m'_1+ m'_2=(m_1+m_2)2^{p}$ in SPDZ is uninfluenced. This is because the SPDZ computation runs over encoded numbers, multiplication operations lead to the expand of precision as $m'_1\times m'_2= m_1 2^{p}\times  m_2 2^{p}
=m_1 m_2 2^{2p}$.
Therefore, it is necessary to keep a precision consistent with a truncation $T(m'_1 m'_2,p)=\lfloor\frac{m'_1 m'_2}{2^{p}}\rceil$ after each multiplication operation, where $T(x,p)=min(max(\lfloor\frac{x}{2^{p}}\rceil,-1+2^p),1-2^p).$

\paratitle{Secure and Verifiable Cross Unit} (Step $\textcircled{3}$):
After the $l$-th pooling layer, given two activation maps $\mathcal{X}_1^l,\mathcal{X}_2^l$ of $\mathcal{N}et_1$ and $\mathcal{N}et_2$, a cross unit yields transferred activation maps $\widetilde{\mathcal{X}}_1^l$ and $\widetilde{\mathcal{X}}_2^l$, the specific computation is shown as

\begin{equation}
\begin{aligned}\label{crossUnit}
\nonumber
\left[\begin{matrix}
\widetilde{x}_1\\ \widetilde{x}_2
\end{matrix}\right]
=
\left[\begin{matrix}
\Theta_{1}\cdot \mathcal{V} \\ \Theta_{2}\cdot \mathcal{V}
\end{matrix}
\right]
=
\left[\begin{matrix}
\theta_{11}, \theta_{12}\\ \theta_{21}, \theta_{22}
\end{matrix}
\right]
\left[\begin{matrix}
{x}_1\\ {x}_2
\end{matrix}
\right].
\end{aligned}
\end{equation}
We describe these computation parameters as follows.
\begin{itemize}
  \item Traversing activation maps $\mathcal{X}_1^l,\mathcal{X}_2^l\leftarrow \mathbb{R}^{h_l\times w_l\times c_l}$, then we obtain ${h_l\times w_l\times c_l}$ vectors $\mathcal{V}=({x}_1,{x}_2)^T$, where the elements ${x}_1\in{{\mathcal{X}}}_1^l$, ${x}_2\in{{\mathcal{X}}}_2^l$ are the location-specific elements.

  \item The degree matrix $\Theta\in[0,1]$ is a symmetric matrix which weighes the degree of shared representations and specified representations. $\Theta_1=(\theta_{11}, \theta_{12})$ is the degree vector of $\mathcal{D}_1$, while $\Theta_2=(\theta_{21}, \theta_{22})$ is the degree vector of $\mathcal{D}_2$.
      Specifically, the higher values of $\theta_{11}$, $\theta_{22}$, the more specified representations of data domains $\mathcal{D}_1$ and $\mathcal{D}_2$. Similarly, the higher values of $\theta_{12}$, $\theta_{21}$ ($\theta_{12}=\theta_{21}$), the more representations are shared between $\mathcal{D}_1$ and $\mathcal{D}_2$.
  \item  The computation results $\widetilde{x}_1\in\widetilde{\mathcal{X}}_1^l$ and $\widetilde{x}_2\in\widetilde{\mathcal{X}}_2^l$ represent the corresponding position in the transferred activation maps.
\end{itemize}

\begin{figure}
  \centering
  \includegraphics[width=3.5in]{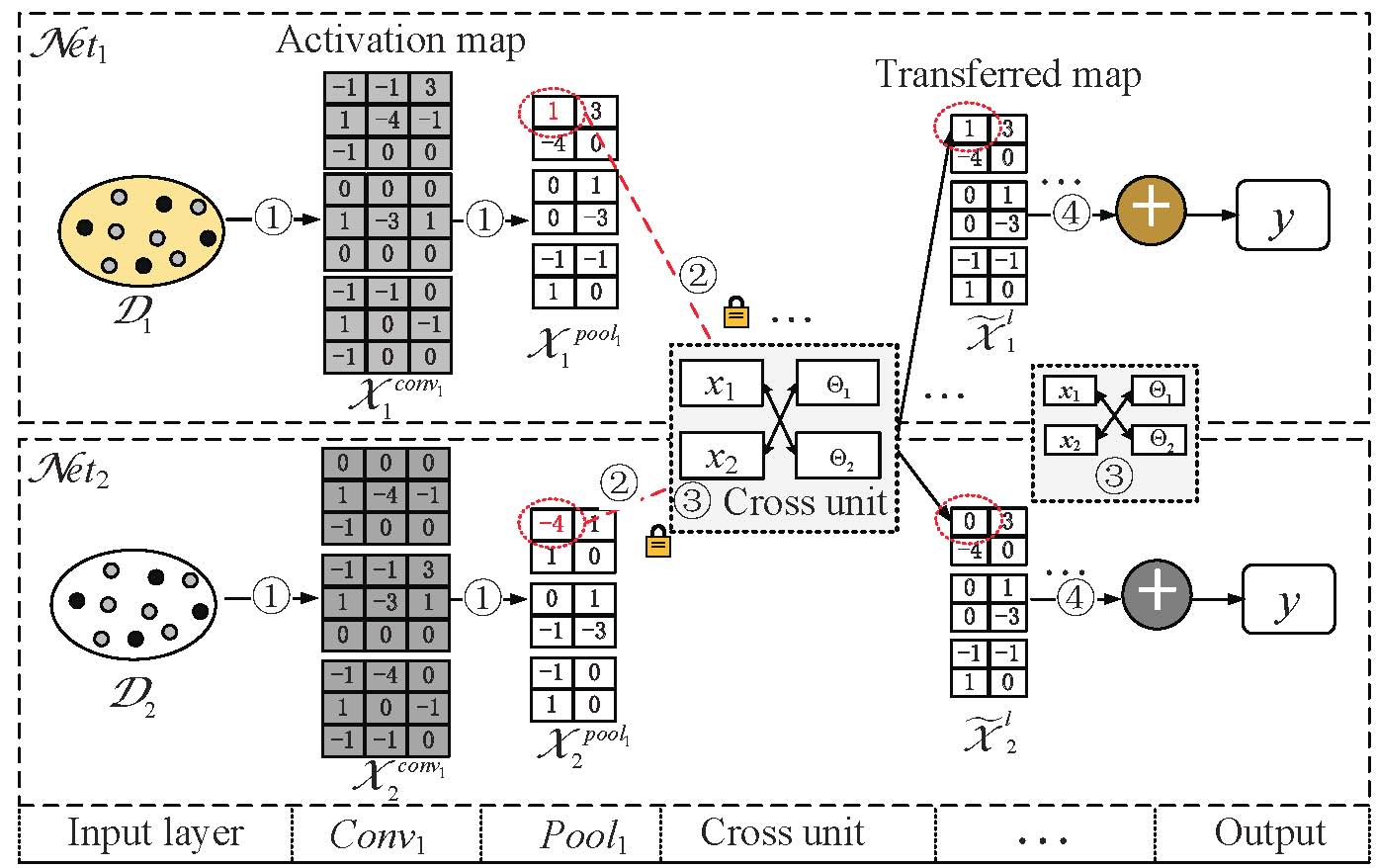}\\
  \caption{An example of cross unit. ``$Conv_1$" means the 1-st convolutional layer, ``$Pool_1$" represents the 1-st pooling layer, $x_1$ and $x_2$ is the location-directed element in activation maps $\mathcal{X}_1^l$ and $\mathcal{X}_2^{l}$, respectively.}\label{crossconnection}
\end{figure}

In the following, we improve SPDZ protocol to securely compute vector multiplication against dishonest majority. The details will be described together with MAC check mechanism.

\paratitle{\rm \underline{\textsf{MACCheck}}}:
Assume that some values of $m$ have been $[\cdot]$-shared and partially opened, then all data domains $\mathcal{D}$ receive additive shares of $m$. However, an adversary can attempt to corrupt the inputs by replacing different and inconsistent attacker-chosen inputs. It is unconfirmed that these shares $[m]$ and the opened value $m$ are correct. Before returning a opened value, it is necessary to verify these shares and opened value with $\gamma(m)=\alpha m$ based on Eq.~\ref{maccheck}, where $\alpha$ is the global MAC key. A MAC-check mechanism is defined in Algorithm 1, which guarantees the correctness of an opened value by verifying the relation among secret shares and MAC shares.

\begin{algorithm}
\scriptsize
\footnotesize
\small
\label{crossunit}
\caption{MACCheck}
\KwIn{Each $\mathcal{D}_i$ has a local MAC key $\alpha^{(i)}$ and $\gamma(m_j)^{(i)}$ and a public set $\{m^{(1)},m^{(2)},...,m^{(n)}\}$}
\KwOut{Success or failure.}
$\mathcal{D}_i$ agrees on a random vector $r\leftarrow \mathbb{Z}_{2^\kappa}^{n}$, and all domains obtain the same vector;\\
$\mathcal{D}_i$ computes the public value $c\leftarrow \sum_{j=1}^n r_j\cdot m^{(j)}$, $\gamma(c)^{(i)}\leftarrow\sum_{i=1}^\mathcal{D} r_j\cdot \gamma(m^{(j)})^{(i)}$ and $\sigma^{(i)}\leftarrow \gamma(c)^{(i)}-\alpha^{(i)} c$;\\
$\mathcal{D}_i$ broadcasts $\sigma^{(i)}$, and all domains receive a set $\{\sigma^{(1)},\sigma^{(2)},...,\sigma^{(n)}\}$;\\
\If {$\Sigma_{i=1}^n \sigma^{(n)} \neq 0$}
{\Return $\bot$ and abort.
}
\end{algorithm}

As a motivating example, given the shares $[m]^{(i)}$, each domain $\mathcal{D}_i$ receives the $i$-th share of $m$ denoted as $[m]^{(i)}$, and agrees on a random vector $r\leftarrow \mathbb{Z}_{2^\kappa}^{n}$. Then, each domain $\mathcal{D}_i$ computes $c\leftarrow \sum_{j=1}^{|\mathcal{D}|} r_j\cdot m^{(j)}$ by adding a random mask to $m^{(j)}$ and $\gamma(c)^{(i)}\leftarrow\sum_{i=1}^{|\mathcal{D}|} r_j\cdot \gamma(m^{(j)})_i$. Finally, all domains implement MAC check with both $c$ and $\gamma(c)_i$. If the MAC check fails, then $\bot$ is outputed and the computation process aborts.
Otherwise, all domains open $m\leftarrow\sum^{|\mathcal{D}|}m^{(i)}$ over $m$ and $\gamma(m^{(i)})$.

\paratitle{\rm \underline{\textsf{VectorMul}}}:
To implement secure multiplication over secret shares of two vectors (\eg $\Theta$ and $\mathcal{V}$),
each data domain $\mathcal{D}_i$ executes SPDZ multiplication operations $\otimes $ over shares $x^{(i)}\in \mathcal{V}^{(i)}$ and $\theta^{(i)}\in \Theta^{(i)}$. Moreover, SPDZ creates and shares tuples $([a],[b],[c])$ in the offline phase, where $c=ab$, and $a,b,c\in\mathbb{Z}_{2^\kappa}$. We illustrate the specific processes as follows.

\begin{itemize}
\item A triple $(a,b,c)$ is involved, where $c=a\cdot b$. Each domain $\mathcal{D}_i$ first operates the masked shares of $\mu^{(i)}=x^{(i)}-a^{(i)}$ and $\nu^{(i)}=\theta^{(i)}-b^{(i)}$ over received shares $a^{(i)}$, $b^{(i)}$, $x^{(i)}$ and $\theta^{(i)}$. Then, the masked shares are broadcasted to all domains. Thus, each domain can open values of $\mu$ and $\nu$. Subsequently, each party $\mathcal{D}_i$ computes $z^{(i)}=c^{(i)}+\mu b^{(i)}+\nu a^{(i)}.$

 \item Besides, to output a share $[z]$, it is required to implement \textsf{MACCheck} to verify both the input and output. If the \textsf{MACCheck} fails, then all domains receive $\bot$ and abort. Otherwise, the final result $z$ of domains is opened with verified shares as
\begin{equation}
\begin{aligned}
\nonumber
z&=\mu\nu+\Sigma_{i=1}^{|\mathcal{D}|}{z_i}=\mu\nu+ c+\mu b+\nu a\\
   &=c+(x-a) b+(\theta-b)a+(x-a)(\theta-b)=x\theta.
\end{aligned}
\end{equation}

\end{itemize}

As a toy example of SPDZ-based vector multiplication, given two vectors $\mathcal{V}=({x}_1, {x}_2)^T$ and $\Theta_1=(\theta_{11}, \theta_{12})$,\footnote{The symbol $``T"$ is the vector transpose.} the result $z$ is produced as
\begin{equation}
\begin{aligned}\label{crossUnit2}
\nonumber
z=\Theta_1\times \mathcal{V}
=
\left[\begin{matrix}
\theta_{11}, \theta_{12}
\end{matrix}
\right]
\left[\begin{matrix}
{x}_1\\ {x}_2
\end{matrix}
\right]
=
\theta_{11}{\otimes }{x}_1+\theta_{12}\otimes {x}_2.
\end{aligned}
\end{equation}

\paratitle{\rm \underline{\textsf{Cross unit {\rm $\Pi_{\text{CrossUnit}}^{\text{SPDZ}}$}}}}:
As both activation maps and degree matrix $\Theta$ include sensitive information of a data domain, SPDZ protocols are adopt to execute secret sharing with a bit of random masks injected into each computation to prevent data leakage between both domains.
All domains quantize elements $x\in \mathcal{V} $ and $\theta\in \Theta_i$ in vectors with $Q(x,p)$ and $Q(\theta,p)$ before implementing secret sharing, then broadcast secret shares $x^{(i)}$, $\theta^{(i)}$ to other data domains, \eg a value $m$ is $[\cdot]$-shared as $[m]=\{m^{(1)},m^{(2)},...,m^{(n)},\alpha^{(1)},\alpha^{(2)},...,\alpha^{(n)},\gamma(m)^{(1)},\gamma(m)^{(2)},...,$ $\gamma(m)^{(n)}\}$, where $\gamma(m)=\alpha m =\sum_{i=1}^n \gamma(m)^{(i)}$.

In a cross unit, it involves vector multiplication as \textsf{VectorMul}$(\Theta_i,{\mathcal{V}})$ to guarantee the correctness of the whole process against covert adversaries, which is demonstrated as
\begin{equation}
\begin{aligned}\label{crossUnit2}
\nonumber
\Theta \times\mathcal{V}
&=
\left[\begin{matrix}
\theta_{11}, \theta_{12}\\ \theta_{21}, \theta_{22}
\end{matrix}
\right]
\left[\begin{matrix}
{x}_1\\ {x}_2
\end{matrix}
\right]
=
\left[\begin{matrix}
\textsf{VectorMul}(\theta_1,{\mathcal{V}})\\ \textsf{VectorMul}(\theta_2,{\mathcal{V}})
\end{matrix}\right]
\\&=
\left[\begin{matrix}
\theta_{11}{x}_1+\theta_{12}{x}_2\\ \theta_{21}{x}_1+\theta_{22}{x}_2
\end{matrix}\right]
=
\left[\begin{matrix}
\widetilde{x}_1\\ \widetilde{x}_2
\end{matrix}\right].
\end{aligned}
\end{equation}
The computation results $\widetilde{x}_i\leftarrow T(\widetilde{x}_i,p)$ returned with a truncation represent the same location in the transferred activation maps $\widetilde{\mathcal{X}}_i^l$.

\paratitle{Tune model}(Step $\textcircled{4}$): To build a transfer model, the global object of collaborative transfer learning is to tune local models $\{\mathcal{N}et_i\}_i^n$ as follows.
\begin{equation}
\begin{aligned}
\nonumber
\mathop{argmin}_{W_1,W_2}\; \mathcal{L}(W_1)+\mathcal{L}(W_2)\\
s.t.\; W_1\in \mathcal{N}et_1, W_2\in \mathcal{N}et_2.
\end{aligned}
\end{equation}

\paratitle{\rm \underline{\textsf{Forward pass}}}: Upon obtaining the transferred $\widetilde{\mathcal{X}}^l_i$, each domain $\mathcal{D}_i$ implements $\widetilde{\mathcal{X}}^{l+1}_i=f(W_i^l \widetilde{\mathcal{X}}^{l}_i)$ and outputs prediction $\mathcal{X}^L$ after $L$ layers.

\paratitle{\rm \underline{\textsf{Backward pass}}}: For minimizing loss function $\mathcal{L}$ measuring the difference between predictions $\mathcal{X}^L$ and ground-truth labels $y$, the objective function can be optimized by the back-propagation method.
During the tune process of trained networks, the global object is divided into local optimization over a single domain $\mathcal{D}_i$. The gradients of data representations are back-propagated, the derivatives of loss function $\mathcal{L}$ in a cross unit are defined as
\begin{equation}
\begin{aligned}\label{crossUnit2}
\left[\begin{matrix}
\nonumber
\frac{\partial \mathcal{L}}{\partial \widetilde{\mathcal{X}}^l_1}\\ \frac{\partial \mathcal{L}}{\partial \widetilde{\mathcal{X}}^l_2}
\end{matrix}\right]
=
\left[\begin{matrix}
\theta_{11}, \theta_{21}\\
\theta_{12}, \theta_{22}
\end{matrix}
\right]
\left[\begin{matrix}
\frac{\partial \mathcal{L}}{\partial \mathcal{X}^l_1}\\ \frac{\partial \mathcal{L}}{\partial \mathcal{X}^l_2}
\end{matrix}
\right], \theta_{12}=\theta_{21}.
\end{aligned}
\end{equation}

\textbf{Remarks}. Compared with previous transfer learning schemes with one-way transfer, we implement collaborative transfer between two data domains with $\Pi_{\text{CrossUnit}}^{\text{SPDZ}}$, which can protect transferred knowledge between two domains, and implement verification to prevent malicious behaviours of certain data domain in the presence of covert adversaries. During the whole process, each data domain is not only a source domain to transfer individual extracted representations, but also a target domain to construct individual CNN model over transferred representations of others.
However, the cross unit only supports the two-domain setting. Thus, it is significant to design a secure and verifiable collaborative transfer scheme under multi-domain setting.

\begin{figure}
  \centering
  \includegraphics[width=3in]{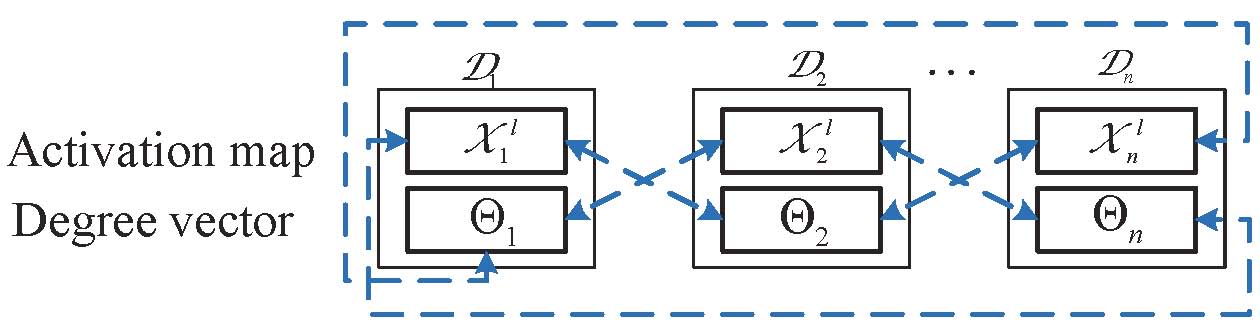}\\
  \caption{The structure of weave unit. }\label{weaveUnit}
\end{figure}

\subsection{Construction of Secure Weave Unit}
To implement the multi-domain transfer learning, we present a weave unit to transfer representations among $n$ domains, where $n-1$ out of $n$ domains can collude with each other.
The key idea is to combine as many data representations as possible to transfer over $n$ data domains.

As depicted in Fig.~\ref{weaveUnit}, the $n\times n$ degree matrix of $\Theta=\{\Theta_i\}_i^n$ (a degree vector $\Theta_i=({\theta_{i1}}, ..., \theta_{in})$ ) is denoted as\footnote{The values of $\theta_{ij}$ and $\theta_{ji}$ are the correlational relationships between $\mathcal{D}_i$ and $\mathcal{D}_j$, \ie $\theta_{ij}=\theta_{ji}$. Thus, $\Theta$ is a symmetric matrix.}

\begin{equation}
\begin{aligned}\label{weaveUnit}
\left[\begin{matrix}
\nonumber
{\theta_{11}}, \theta_{12}, ..., \theta_{1n} \\ \theta_{21}, {\theta_{22}}, ..., \theta_{2n}\\ ...\\ \theta_{n1}, \theta_{n2}, ..., {\theta_{nn}}
\end{matrix}
\right].
\end{aligned}
\end{equation}
The elements ${\theta_{11}}$, ${\theta_{22}}$,..., ${\theta_{nn}}$ in the diagonal line of the matrix are denoted as $\theta_{\text{s}}$ that describes the degree of specified representations over individual data domains. The other elements are denoted as $\theta_{\text{t}}$ that describes the degree of transferred representations over other data domains. To implement flexible transfer learning, $\theta_{t_{ij}}$ is defined as the degree of transferred representations between $\mathcal{D}_i$ and $\mathcal{D}_j$, a higher value of $\theta_{t_{ij}}$ means that more data representations of $\mathcal{D}_j$ are transferred to $\mathcal{D}_i$, where $\theta_{t_{ij}}=\theta_{t_{ji}}$.

\begin{figure*}[!ht]
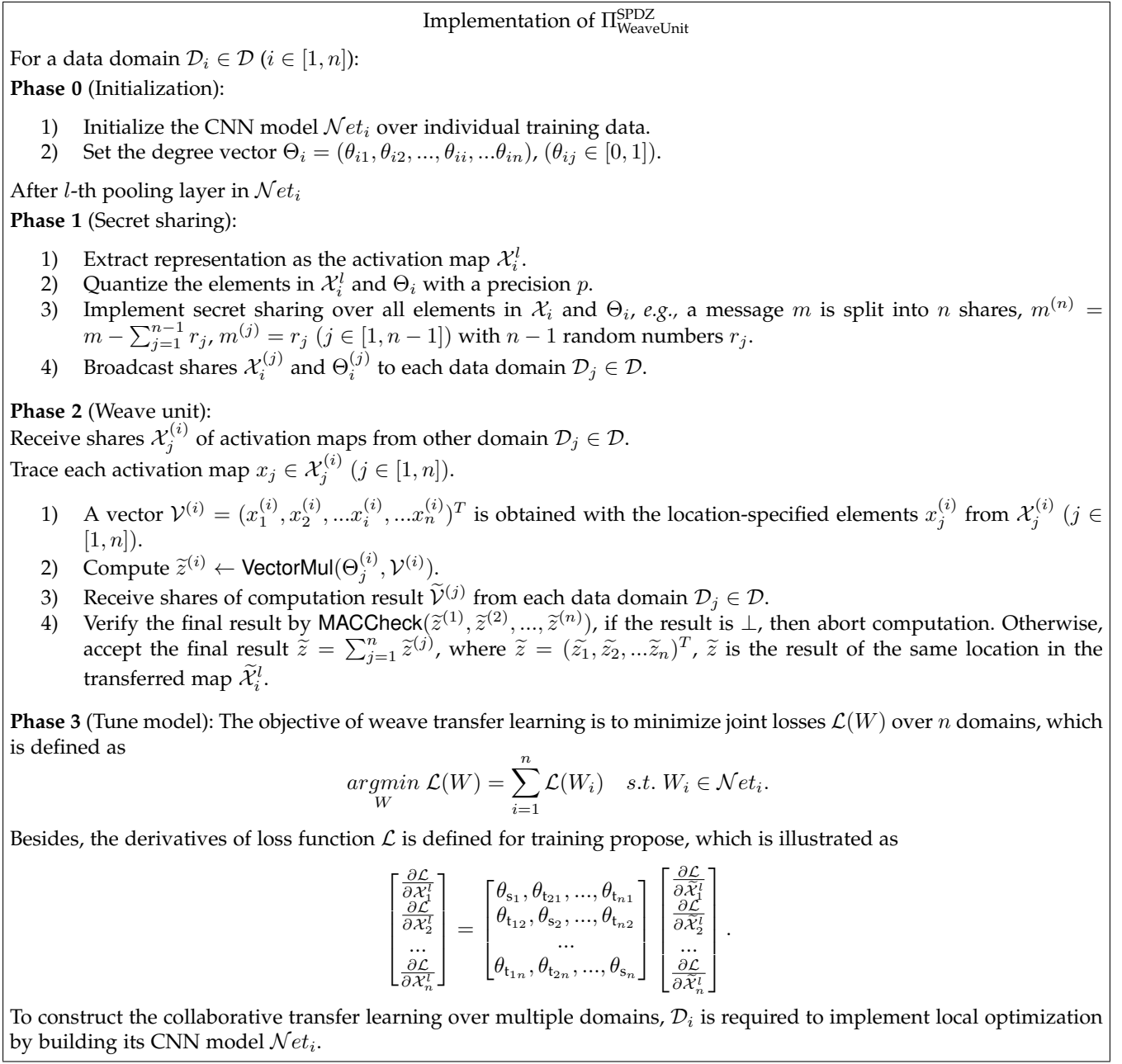

\centering
\begin{boxedminipage}{0.98\textwidth}
\begin{center}
{Implementation of $\Pi_{\text{WeaveUnit}}^{\text{SPDZ}}$}\\
\end{center}
For a data domain $\mathcal{D}_i \in \mathcal{D}$ ($i\in[1,n]$):

{\paratitle{Phase 0} (Initialization)}:
\begin{enumerate}
  \item Initialize the CNN model $\mathcal{N}et_i$ over individual training data.
  \item Set the degree vector $\Theta_i=(\theta_{i1},\theta_{i2},...,\theta_{ii},...\theta_{in})$, $(\theta_{ij}\in[0,1])$.
\end{enumerate}
After $l$-th pooling layer in $\mathcal{N}et_i$

\paratitle{Phase 1} (Secret sharing):
\begin{enumerate}
 \item Extract representation as the activation map $\mathcal{X}_i^l$.
 \item Quantize the elements in $\mathcal{X}_i^l$ and $\Theta_i$ with a precision $p$.
 \item Implement secret sharing over all elements in $\mathcal{X}_i$ and $\Theta_i$, \eg a message $m$ is split into $n$ shares, $m^{(n)}=m-\sum_{j=1}^{n-1}r_j$, $m^{(j)}=r_j$ $(j\in [1,n-1])$ with $n-1$ random numbers $r_j$.
 \item Broadcast shares $\mathcal{X}_i^{(j)}$ and $\Theta_i^{(j)}$ to each data domain $\mathcal{D}_j \in \mathcal{D}$.
 \end{enumerate}
 \paratitle{Phase 2} (Weave unit):

 Receive shares $\mathcal{X}_j^{(i)}$ of activation maps from other domain $\mathcal{D}_j \in \mathcal{D}$.

 Trace each activation map $x_j\in\mathcal{X}_j^{(i)}$ $(j\in[1,n])$.

 \begin{enumerate}
 \item 
 A vector $\mathcal{V}^{(i)}=(x_1^{(i)},x_2^{(i)},...x_i^{(i)},...x_n^{(i)})^T$ is obtained with the location-specified elements $x_j^{(i)}$ from $\mathcal{X}_j^{(i)}$ $(j\in[1,n])$.
 \item Compute $\widetilde{z}^{(i)}\leftarrow \textsf{VectorMul}(\Theta_j^{(i)},\mathcal{V}^{(i)})$.
 \item Receive shares of computation result $\widetilde{\mathcal{V}}^{(j)}$ from each data domain $\mathcal{D}_j \in \mathcal{D}$.
 \item Verify the final result by $\textsf{MACCheck}(\widetilde{z}^{(1)},\widetilde{z}^{(2)},...,\widetilde{z}^{(n)})$, if the result is $\bot$, then abort computation. Otherwise, accept the final result $\widetilde{z}=\sum_{j=1}^n \widetilde{z}^{(j)}$, where $\widetilde{z}=(\widetilde{z_1},\widetilde{z_2},...\widetilde{z}_n)^T$, $\widetilde{z}$ is the result of the same location in the transferred map $\widetilde{\mathcal{X}}_i^l$.
 \end{enumerate}
 \paratitle{Phase 3} (Tune model):
 The objective of weave transfer learning is to minimize joint losses $\mathcal{L}(W)$ over $n$ domains, which is defined as
\begin{equation}
\begin{aligned}
\nonumber
\mathop{argmin}_{W}\; &\mathcal{L}(W)=\sum_{i=1}^n \mathcal{L}(W_i)
&s.t.\; W_i\in \mathcal{N}et_i.
\end{aligned}
\end{equation}
Besides, the derivatives of loss function $\mathcal{L}$ is defined for training propose, which is illustrated as
\begin{equation}
\begin{aligned}\label{crossUnit2}
\nonumber
\left[\begin{matrix}
\frac{\partial \mathcal{L}}{\partial \mathcal{X}_1^l}\\ \frac{\partial \mathcal{L}}{\partial \mathcal{X}_2^l}\\...\\ \frac{\partial \mathcal{L}}{\partial \mathcal{X}_n^l}
\end{matrix}\right]
=
\left[\begin{matrix}
\theta_{\text{s}_1}, \theta_{\text{t}_{21}},...,\theta_{\text{t}_{n1}} \\
\theta_{\text{t}_{12}}, \theta_{\text{s}_2}, ..., \theta_{\text{t}_{n2}} \\
...\\
\theta_{\text{t}_{1n}}, \theta_{\text{t}_{2n}}, ..., \theta_{\text{s}_n}
\end{matrix}
\right]
\left[\begin{matrix}
\frac{\partial \mathcal{L}}{\partial \widetilde{\mathcal{X}}_1^l}\\ \frac{\partial \mathcal{L}}{\partial \widetilde{\mathcal{X}}_2^l}\\...\\ \frac{\partial \mathcal{L}}{\partial \widetilde{\mathcal{X}}_n^l}
\end{matrix}\right].
\end{aligned}
\end{equation}

To construct the collaborative transfer learning over multiple domains, $\mathcal{D}_i$ is required to implement local optimization by building its CNN model $\mathcal{N}et_i$.

\end{boxedminipage}
\caption{Detailed descriptions of $\Pi_{\text{WeaveUnit}}^{\text{SPDZ}}$.}
\label{schemedefinition}
\end{figure*}

Specially, a weave unit is defined as
\begin{equation}
\begin{aligned}\label{crossUnit}
\nonumber
\left[\begin{matrix}
\widetilde{x}_1\\ \widetilde{x}_2\\...\\ \widetilde{x}_n
\end{matrix}\right]
=
{\left[\begin{matrix}
{\theta_{\text{s}_1}}, \theta_{\text{t}_{12}}, ..., \theta_{\text{t}_{1n}} \\ \theta_{\text{t}_{21}}, {\theta_{\text{s}_2}}, ..., \theta_{\text{t}_{2n}}\\ ...\\ \theta_{\text{t}_{n1}}, \theta_{\text{t}_{n2}}, ..., {\theta_{\text{s}_n}}
\end{matrix}
\right]}
\cdot
\left[\begin{matrix}
{x}_1\\ {x}_2\\...\\ {x}_n
\end{matrix}
\right],
\end{aligned}
\end{equation}
where $\mathcal{V}=({x}_1,{x}_2,...,{x}_n)$, ${x}_i\in \mathcal{X}_i^l$ are the elements of the same location in activation maps, $i\in [1,n]$.
$\widetilde{x}_1$ is the result of the corresponding position in the transferred activation map $\widetilde{\mathcal{X}_1^l}$, which is computed as
\begin{equation}
\begin{aligned}
\nonumber
\widetilde{x}_1=\Theta_1\cdot \mathcal{V}=\theta_{\text{s}_1} {x}_1+\theta_{\text{t}_{12}}{x}_2+\theta_{\text{t}_{13}}{x}_3+...+\theta_{\text{t}_{1n}}x_n.
\end{aligned}
\end{equation}
A specified weave transferred result $\widetilde{x}_1$ is determined by $\theta_{\text{s}}$ and $\theta_{\text{t}}$. With a higher value of $\theta_{\text{s}}$, the trained $\mathcal{N}et$ focuses on more data representations from individual images. With a higher value of $\theta_{\text{t}}$, $\mathcal{N}et$ is tuned over more transferred data representations from other domains.

\paratitle{Secure and Verifiable Weave Unit} $\Pi_{\text{WeaveUnit}}^{\text{SPDZ}}$:
To guarantee the security in the presence of dishonest majority, the algorithm $\Pi_{\text{WeaveUnit}}^{\text{SPDZ}}$ is presented. The whole process involves the predefined \textsf{MACCheck} and SPDZ multiplication, which is divided into several phases in Fig.~\ref{schemedefinition}.

\section{Security Analysis} \label{section:Security Analysis}
In this section, we first give the security definition, and then analyze our proposed scheme to evaluate whether it satisfies the privacy requirements in Section~\ref{sec:threatmodel} under the following security definitions.
\subsection{Security Definition}
We follow the security definition formalized in~\cite{hazay2010note,furukawa2017high}, the security of a protocol $\pi$ is defined as the indistinguishability between the real-model executed by an adversary $\mathcal{A}$ and an ideal functionality with a simulator $\mathcal{S}$, which is formalized as $\textsf{REAL}_{\pi,\mathcal{A}}\stackrel{c}{\equiv}\textsf{IDEAL}_{\pi,\mathcal{S}}$.

\paratitle{Real-world model {\rm $\textsf{REAL}$}}: The $n$-party protocol $\pi$ is executed over data domains $\mathcal{D}$. Each data domain $\mathcal{D}_i$ provides the public inputs $\textsf{Input}_i^p=(\mathcal{X}_i^p$, $\Theta_i^p)$ and secret inputs $\textsf{Input}_i^s=(\mathcal{X}_i^s$, $\Theta_i^s)$, then the public output $\textsf{Output}_i^p$ and secret output $\textsf{Output}_i^s$ are produced with random masks $r_i\in \mathbb{Z}_{2^\kappa}$ ($i\in[1,n]$). Besides, there exist some subsets of multiple independent covert adversaries $\{\mathcal{A}_1, \mathcal{A}_2,...,\mathcal{A}_n\}$, where $\mathcal{A}_i$ can corrupt a data domain $\mathcal{D}_i$, and the number of adversaries can be a majority of domains.

Here, let $\mathcal{D}^{{c}}\subset\mathcal{D}$ be the corrupted data domains and $\mathcal{D}^{{h}}\subset\mathcal{D}$ be the honest data domains, where $\mathcal{D}=\mathcal{D}^{{c}}\cup\mathcal{D}^{{h}}$. In $\textsf{REAL}$, with the given inputs, the output of the protocol $\pi$ after a real-model execution is defined as follows:
\begin{equation}
\begin{aligned}
\nonumber
\textsf{REAL}_{\pi,\mathcal{A}}=\{\textsf{REAL}_{\pi,\mathcal{A}_i}(\mathcal{D}^c,\kappa,\mathcal{{X}}_i,\Theta_i,r_i)\}_{i\in[1,n]},
\end{aligned}
\end{equation}
where $\kappa$ is security parameter, $\mathcal{{X}}=\{\mathcal{{X}}_1,...,\mathcal{{X}}_n\}$ and $\Theta=\{\Theta_1,...,\Theta_n\}$ are the set of activation maps and the set of degree vectors from all data domains, respectively.

\paratitle{Ideal-world model {\rm$\textsf{IDEAL}$}}: The function $f$ is executed as a probabilistic $n$-party function in Probabilistic Polynomial Time (PPT), which is defined as $f(\kappa,\textsf{Input}_1^s,\textsf{Input}_1^p,...,\textsf{Input}_n^s,\textsf{Input}_n^p,r)$, and $r$ is a set of random masks. In $\textsf{IDEAL}$, all domains send individual inputs to a trusted third party $\mathcal{T}$ that executes $f$ over these inputs and returns $(\textsf{Output}_i^p,\textsf{Output}_i^s)$ to $\mathcal{D}_i$. After an ideal-model execution with the presence of PPT simulators $\mathcal{S}_i$ ($i\in[1,n]$), the view is defined as
\begin{equation}
\begin{aligned}
\nonumber
\textsf{IDEAL}_{f,\mathcal{S}}=\{\textsf{IDEAL}_{f,\mathcal{S}_i}(\mathcal{D}^c,\kappa,\mathcal{{X}}_i,\Theta_i,r_i)\}_{i\in[1,n]}.
\end{aligned}
\end{equation}

\paratitle{Hybrid model {\rm$\textsf{HYB}$}}: In the $(g_1,g_2,...,g_l)$-hybrid model, the protocol $\pi$ is executed in the real-world model, except that data domains access to the trusted third party $\mathcal{T}$ for evaluating $n$-party functions $g_1,g_2,...,g_l$, while these ideal evaluations are executed in the ideal-world model.
\begin{equation}
\begin{aligned}
\nonumber
\textsf{HYB}_{\pi,\mathcal{A}}^{g_1,g_2,...,g_l}=
\{\textsf{HYB}_{\pi,\mathcal{A}_i}^{g_1,g_2,...,g_l}(\mathcal{D}^c,\kappa,\mathcal{{X}}_i,\Theta_i,r_i)\}_{i\in[1,n]}.
\end{aligned}
\end{equation}
Here, the security of a protocol $\pi$ is required with the real-world execution or a $(g_1,g_2,...,g_l)$-hybrid execution of an ideal function $f$ without leaking any sensitive information to $\mathcal{{A}}$.

Based on above models, the formal security definition is provided as follows.

\begin{definition}\rm
The $n$-party protocol $\pi$ can securely implement $n$-party function $f$ in a $(g_1,g_2,...,g_l)$-hybrid model with an adversary $\mathcal{{A}}$ that can corrupt a subset of $\mathcal{{D}}^{h}$, there exists an ideal-model simulator $\mathcal{S}$ such that
\begin{equation}
\begin{aligned}
\nonumber
\textsf{IDEAL}_{f,\mathcal{S}}(\mathcal{D}^c,\kappa,\mathcal{{X}},\Theta,r)\stackrel{c}{\equiv}
\textsf{REAL}_{f,\mathcal{A}}(\mathcal{D}^c,\kappa,\mathcal{{X}},\Theta,r).
\end{aligned}
\end{equation}
\end{definition}

\subsection{Security Proofs}
Based on above security definition, we theoretically prove that the security of the proposed system is computationally indistinguishable in following theorems.

\begin{theorem}\rm
Let the secure cross unit be a protocol that securely computes a functionality $\Pi_{\text{CrossUnit}}^{\text{SPDZ}}$ between two data domains ($\mathcal{{D}}_1$ and $\mathcal{{D}}_2$) in the presence of a covert adversary.
\end{theorem}
\begin{proof}
We consider the case with a semi-honest adversary and a malicious adversary, respectively.

\paratitle{Semi-honest setting}.
Here, we respectively analyze the security under two following settings with above real \textit{vs.} ideal model.

\textit{One semi-honest party $(\mathcal{{D}}_1^{\text{c}}, \mathcal{{D}}_2^{\text{h}})$}:
In this setting, $\mathcal{{D}}_1$ is semi-honest that follows the protocol but may try to learn private information of $\mathcal{{D}}_2$, and the honest $\mathcal{{D}}_2$ is denoted as $\mathcal{{D}}_2^{\text{h}}$.
Here, the simulator $\mathcal{{S}}_1$ of the adversary $\mathcal{{A}}_1$ is constructed to play the role of $\mathcal{{D}}_1$ by interacting with the other domain $\mathcal{{D}}_2$.
During the process of $\Pi_{\text{CrossUnit}}^{\text{SPDZ}}$, each domain implements secret sharing over individual activation maps $\mathcal{{X}}$ and the degree vector $\Theta_i$, and then these shares are broadcasted to other domains as the inputs of $\Pi_{\text{CrossUnit}}^{\text{SPDZ}}$. Since the values of these shares are masked with random numbers $r\leftarrow \mathbb{Z}_{2^\kappa}$, \eg a value $x=\sum_{i=1}^2 x^{(i)}  \mod {2^{\kappa}}$, where $x^{(1)}=x-r$, $x^{(2)}=r$, the actual values of the inputs cannot be recovered with the protection of random mask $r$. To implement \textsf{VectorMul} over shares of activation maps $\mathcal{{X}}$ and degree vector $\Theta_i$, each domain performs local computation of addition and multiplication operations over these shares, these local computation results $[z]$ are broadcasted to open the final result. As the computation parameters  $[z]$ are still masked with random numbers, once $\mathcal{{A}}_1$ receives these intermediate parameters, it is still impossible to obtain the actual values of intermediate parameters $[z]$.
It is obvious that the views of the semi-honest adversary $\mathcal{{A}}_1$ are indistinguishable in both real and ideal model, as represented in
\begin{equation}
\begin{aligned}\label{secyrity_fl}
\nonumber
\textsf{IDEAL}_{\mathcal{S}_1}(\mathcal{{X}},\Theta,\mathcal{{D}})\stackrel{c}{\equiv} \textsf{REAL}_{\mathcal{A}_1}(\mathcal{{X}},\Theta,\mathcal{{D}}).
\end{aligned}
\end{equation}
In the same way, it is proved the security in the setting of $(\mathcal{{D}}_1^{\text{h}}, \mathcal{{D}}_2^{\text{c}})$.

\textit{Two semi-honest parties $(\mathcal{{D}}_1^{\text{c}}, \mathcal{{D}}_2^{\text{c}})$}:
In this setting, both $\mathcal{{D}}_1$ and $\mathcal{{D}}_2$ are semi-honest, the simulators $\mathcal{{S}}_1$ and $\mathcal{{S}}_2$ are constructed to play the roles of $\mathcal{{D}}_1$ and $\mathcal{{D}}_2$, respectively. During the process of interacting with the other domain in $\Pi_{\text{CrossUnit}}^{\text{SPDZ}}$, all the inputs and intermediate parameters are $[\cdot]$-shared before being broadcasted to the other domain. The privacy of the inputs and intermediate results can be protected as the views of both $\mathcal{{A}}_1$ and $\mathcal{{A}}_2$ are indistinguishable between the real and ideal model, which is represented as
\begin{equation}
\begin{aligned}\label{secyrity_fl}
\nonumber
\textsf{IDEAL}_{\mathcal{S}_1,\mathcal{S}_2}(\mathcal{{X}},\Theta,\mathcal{{D}})\stackrel{c}{\equiv} \textsf{REAL}_{\mathcal{A}_1,\mathcal{A}_2}(\mathcal{{X}},\Theta,\mathcal{{D}}).
\end{aligned}
\end{equation}
Based on the above analysis, we conclude that our protocol $\Pi_{\text{CrossUnit}}^{\text{SPDZ}}$ can securely implement under the setting of $(\mathcal{{D}}_1^{\text{c}}, \mathcal{{D}}_2^{\text{h}})$ and $(\mathcal{{D}}_1^{\text{c}}, \mathcal{{D}}_2^{\text{c}})$, which  satisfies privacy requirements of semi-honest adversarial model $\text{Adv}$.

\paratitle{Malicious setting}.
Then, we denote $f_{\text{SPDZ}}$ for the security analysis with a malicious adversary $\mathcal{A}$. $f_{\text{SPDZ}}$ is an ideal function that implements the SPDZ-based cross unit $\Pi_{\text{CrossUnit}}^{\text{SPDZ}}$. Let a secure cross unit be a protocol that securely computes a protocol $\Pi_{\text{CrossUnit}}^{\text{SPDZ}}$ in the $(f_{\text{SPDZ}})$-hybrid model between two data domains ($\mathcal{{D}}_1$ and $\mathcal{{D}}_2$) against a malicious adversary statically corrupting $n-1$ out of $n$ data domains.

To prove the security of $\Pi_{\text{CrossUnit}}^{\text{SPDZ}}$ in the $(f_{\text{SPDZ}})$-hybrid model, we construct a simulator $\mathcal{{S}}$ to  prove that the simulator's view is indistinguishable from the view of real-world model, the specified process is demonstrated as follows.

1) $\mathcal{{S}}$ extracts the local activation maps after the building of a CNN model. These activation maps are adopted for interactive data representations between two data domains.

2) $\mathcal{{S}}$ simulates the pre-process phase by receiving the inputs (\ie activation maps $\mathcal{{X}}_i$ and a degree vector $\Theta_i$) from the adversary $\mathcal{{A}}$, and generates additive shares before broadcasting them to the other domain.

3) $\mathcal{{S}}$ executes the ideal functionality \textsf{VectorMul} over these inputs from $\mathcal{{A}}$. In \textsf{VectorMul}, $\mathcal{{S}}$ simulates the honest parties with correct computation of vector multiplication over secret shares. During these phases, $\mathcal{{S}}$ simulates several times of SPDZ multiplication ``$\otimes$", where all inputs and intermediate parameters are masked with statistically indistinguishable uniformly random numbers $r\in \mathbb{Z}_{2^\kappa}$. Thus, the real distributions of these inputs and intermediate parameters in the simulator $\mathcal{{S}}$ are statistically indistinguishable from the view of $\textsf{REAL}$.

4) $\mathcal{{S}}$ opens the final result over shares with the \textsf{MACCheck} mechanism by simulating $f_{\text{SPDZ}}$. $\mathcal{{S}}$ receives the global MAC key $\alpha$, then splits $\alpha$ into two random shares, one of which is sent to the other domain. For each input share $[m]^{(i)}$, $\mathcal{{S}}$ generates MAC shares $\gamma(m)^{(i)}$ to verify whether secret shares and MAC shares satisfy the invariant $\alpha (\sum m^{(i)})- \sum \gamma(m)^{(i)}$.
The malicious data domain also provides additive shares and MAC shares for the verification of a final result. If the validation fails, then \textsf{MACCheck} aborts the computation. Otherwise, $\mathcal{{S}}$ obtains the final result for the adversary $\mathcal{A}$.

5) $\mathcal{{S}}$ follows the training process by tuning local model for the construction of a CNN.

Based on the above analysis, the view of an adversary $\mathcal{A}$ is indistinguishable between \textsf{IDEAL} and \textsf{REAL} with the underlying SPDZ computation, which is represented as $\textsf{IDEAL}_{\Pi_{\text{CrossUnit}}^{\text{SPDZ}},\mathcal{S}}(\mathcal{{X}},\Theta,\mathcal{{D}})\stackrel{c}{\equiv}
\textsf{REAL}_{\Pi_{\text{CrossUnit}}^{\text{SPDZ}},\mathcal{A}}(\mathcal{{X}},\Theta,\mathcal{{D}}).$
\end{proof}

\begin{theorem}\rm
Let the secure weave unit be a protocol that securely computes a functionality $\Pi_{\text{WeaveUnit}}^{\text{SPDZ}}$ among multiple data domains in the presence of covert adversaries.
\end{theorem}

\begin{proof}
We separately analyze $\Pi_{\text{WeaveUnit}}^{\text{SPDZ}}$ with the semi-honest setting and malicious setting.

\paratitle{Semi-honest setting}. Let $\mathcal{A}$ be an augmented semi-honest adversary and $\mathcal{S}$ be a simulator that is guaranteed to the security of $\Pi_{\text{WeaveUnit}}^{\text{SPDZ}}$~\cite{furukawa2017high}.
We construct the simulator $\mathcal{S}$ can do everything what a data domain $\mathcal{D}_i$ can do. Here, it is an extension of \textbf{Theorem} 1 under the multi-domain settings.

\paratitle{Malicious setting}.
In the malicious ideal-model, each data domain $\mathcal{D}_i$ holds an activation map $\mathcal{X}_i$ ($\mathcal{D}_i\in \mathcal{D}$). There is a PPT simulator $\mathcal{S}$ can select and change the inputs for a corrupted data domain, the main idea is that $\mathcal{S}$ executes a series of modifications to our protocol. In our $\textsf{HYB}$ model, $\textsf{hyb}$ denotes a modification to the predefined protocol, the specified process is shown as follows.

In the $\textsf{hyb}$, $\mathcal{S}$ changes all secret shares (\ie shares of an activation map $\mathcal{X}_i$ and a degree vector $\Theta_i$) sent by honest data domains to other domains with shares of random values. It is obvious that the adversary $\mathcal{A}$ cannot learn extra shares of MAC key $\alpha^{(i)}$. However, all honest domains execute $\textsf{MACCheck}$ over shares $x^{(i)}$ of the final result and corresponding MAC key shares $\alpha_i$ in \textbf{Phase 2}, where $x=\sum x^{(i)} \mod {2^{\kappa}}$, $\alpha=\sum \alpha^{(i)}$ and $\gamma(x)_i=\alpha_i x^{(i)} $. Each data domain verifies the opened result $x$ with $\textsf{MACCheck}$ by judging if this is true:
  \begin{equation}
\begin{aligned}
\nonumber
\gamma(x)=\sum \gamma(x)_i =\sum \alpha^{(i)} x^{(i)} \mod {2^{\kappa}}=\alpha x.
\end{aligned}
\end{equation}
Since $\mathcal{S}$ can only change the content of secret shares, it cannot modify an additive share of the corresponding MAC value $\gamma(x)$, which is computed by using additive shares of MAC key $\alpha_i$ on each data domain $\mathcal{D}_i$.
\textsf{MACCheck} can enable each domain to correctly compute in the weave unit. If the inputs and opened values don't pass the MAC check, then all domains will receive ``$\bot$" and abort computation. Otherwise, the final computation results are returned to the weave unit.

Therefore, the simulator $\mathcal{S}$ has completed the simulation process, where $\mathcal{S}$ successfully simulates $\textsf{IDEAL}$ without leaking original values of activation maps $\mathcal{X}$ and degree vectors $\Theta$ for all data domains $\mathcal{D}_i\in \mathcal{D}$.
Thus, it indicates the indistinguishability between this hybrid $\textsf{HYB}$ and real model $\textsf{REAL}$ based on above analysis, which is represented as $\textsf{IDEAL}_{\Pi_{\text{WeaveUnit}}^{\text{SPDZ}},\mathcal{S}}(\mathcal{X},\Theta,\mathcal{D})\stackrel{c}{\equiv}
\textsf{REAL}_{\Pi_{\text{WeaveUnit}}^{\text{SPDZ}},\mathcal{A}}(\mathcal{X},\Theta,\mathcal{D}).$
\end{proof}

\section{Performance Evaluation} \label{section:Performance Evaluation}
In this section, we discuss experimental setup and evaluate VerifyTL on real-world datasets, and we compare VerifyTL with existing solutions.
\subsection{Experimental Setup}
The experiments were executed in Java and were implemented on a six-core 2.80GHz machine with Inter i5-8400H processor, 16GB RAM, running Ubuntu, and VerifyTL is evaluated in parallel.
We begin the experiments by introducing training datasets. The communication among different data domains relies on TCP with authenticated channels (through TLS).

\textbf{Datasets}. We evaluated our methods over two different real-world datasets.
\begin{itemize}
  \item MNIST\footnote{http://yann.lecun.com/exdb/mnist}. MNIST contains 60K training samples and 10K test samples. Each sample is a grayscale of 10 different handwritten digits formatted as $28*28$ images.

  \item Fashion MNIST\footnote{https://www.kaggle.com/zalando-research/fashionmnist}. The size of training data is 60K, and the size of test data is 10K, while a fashion
MNIST instance is a $28*28$ image contains 10 labels as ``T-shirt", ``trouser", ``pullover", ``dress", ``coat", ``sandal", ``shirt", ``sneaker", ``bag", and ``boot".
\end{itemize}

\begin{figure*}[!ht]
	  \centering
\subfigure[]{\includegraphics[width=1.4in]{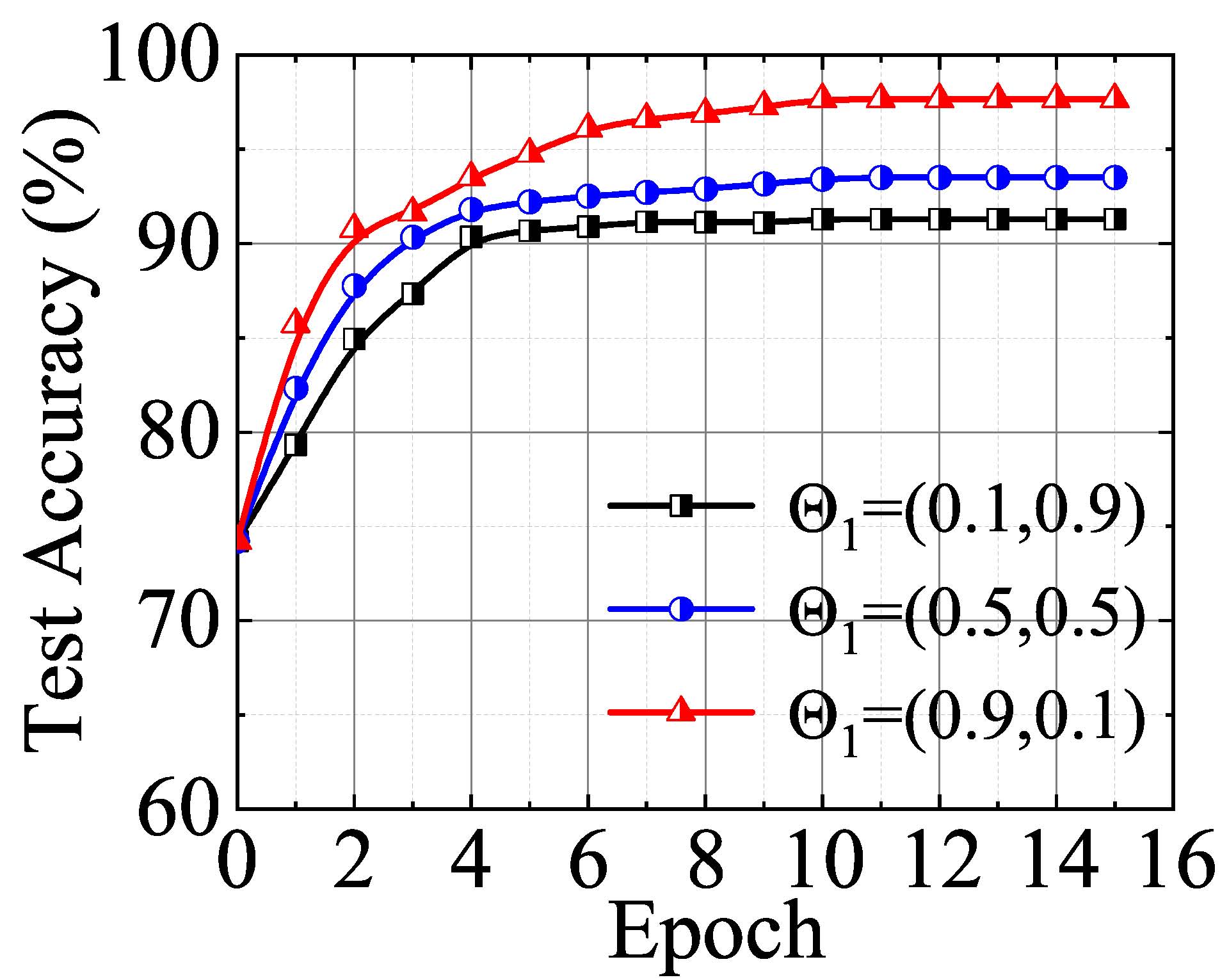}}
\subfigure[]{\includegraphics[width=1.4in]{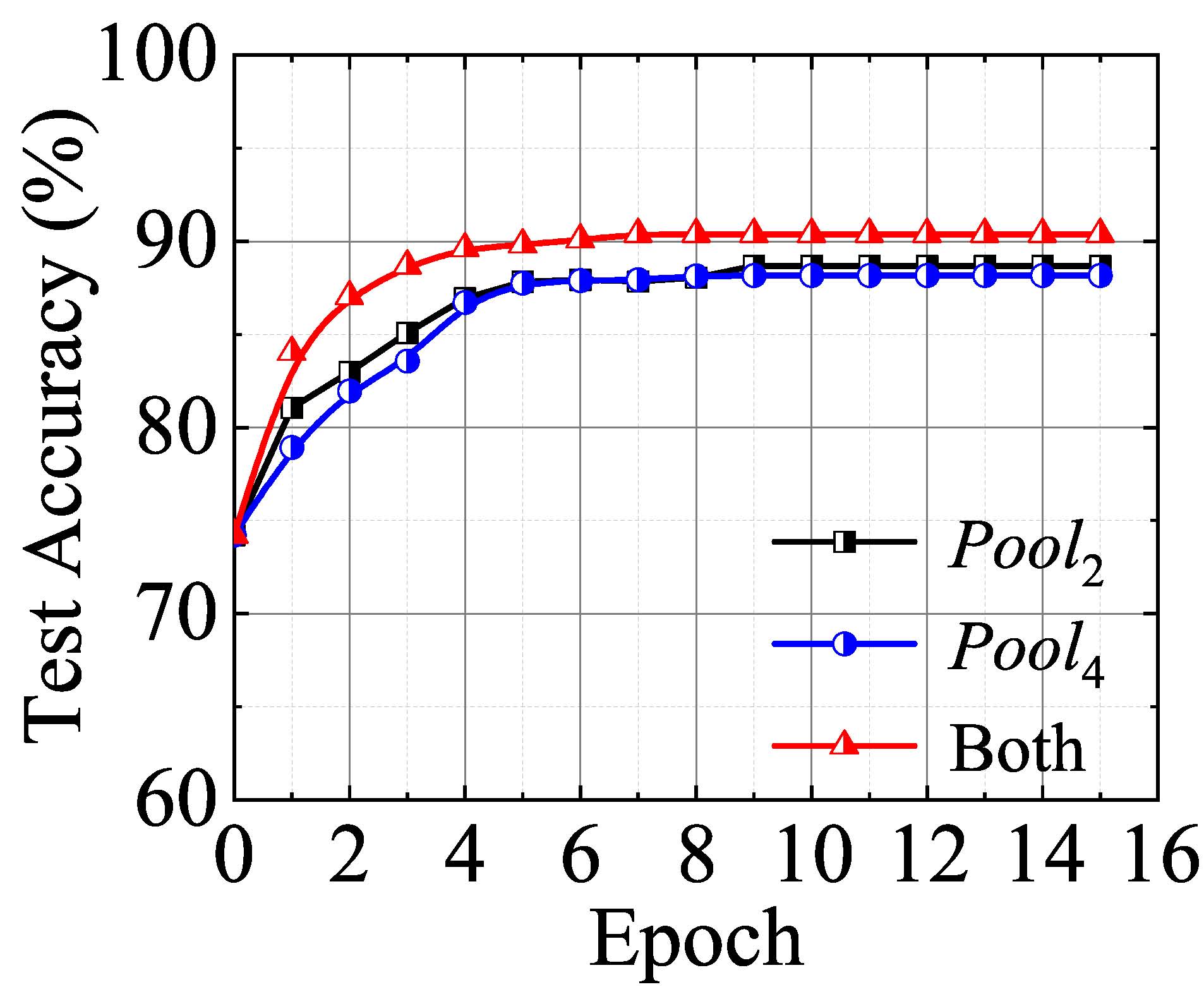}}
\subfigure[]{\includegraphics[width=1.4in]{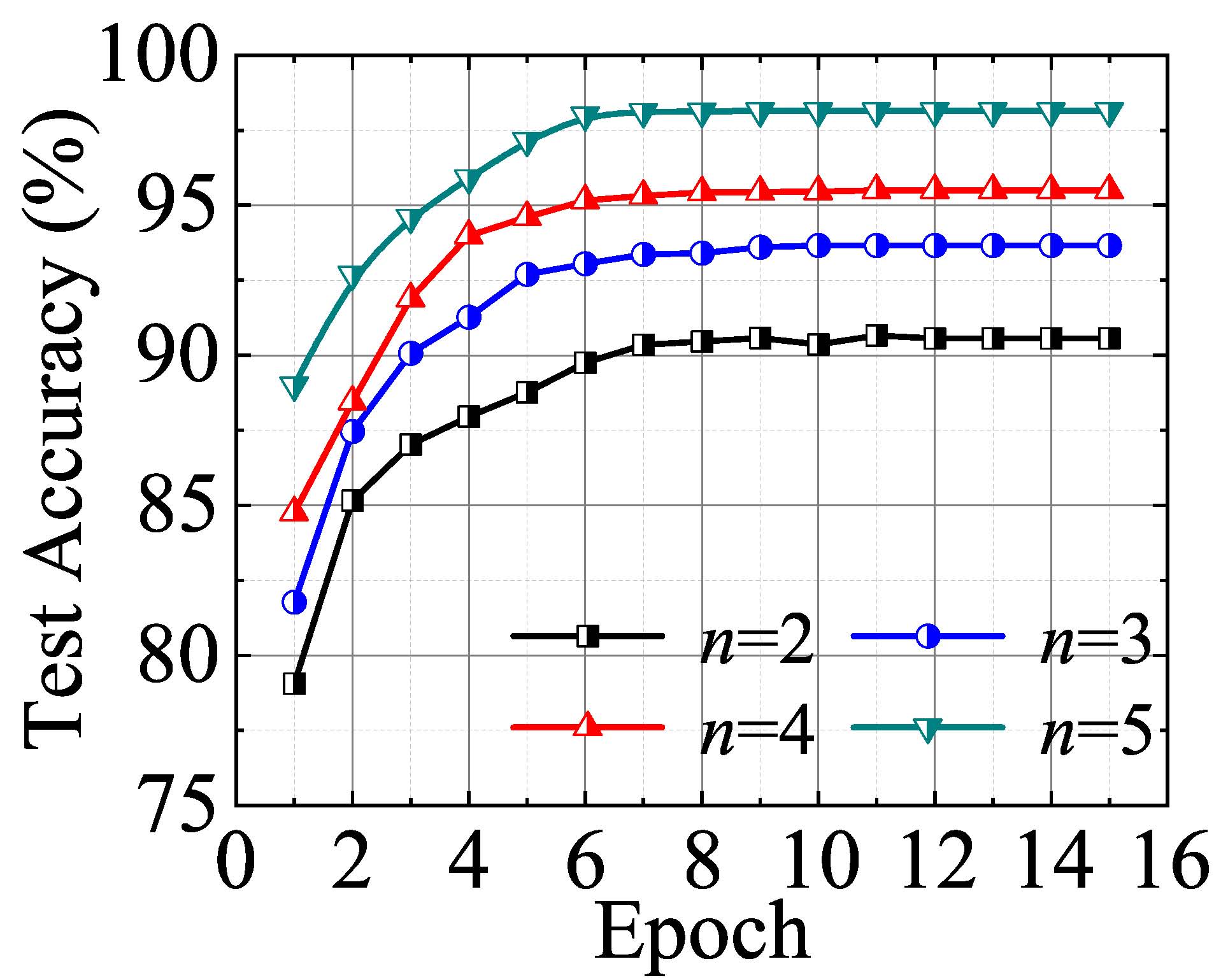}}
\subfigure[]{\includegraphics[width=1.4in]{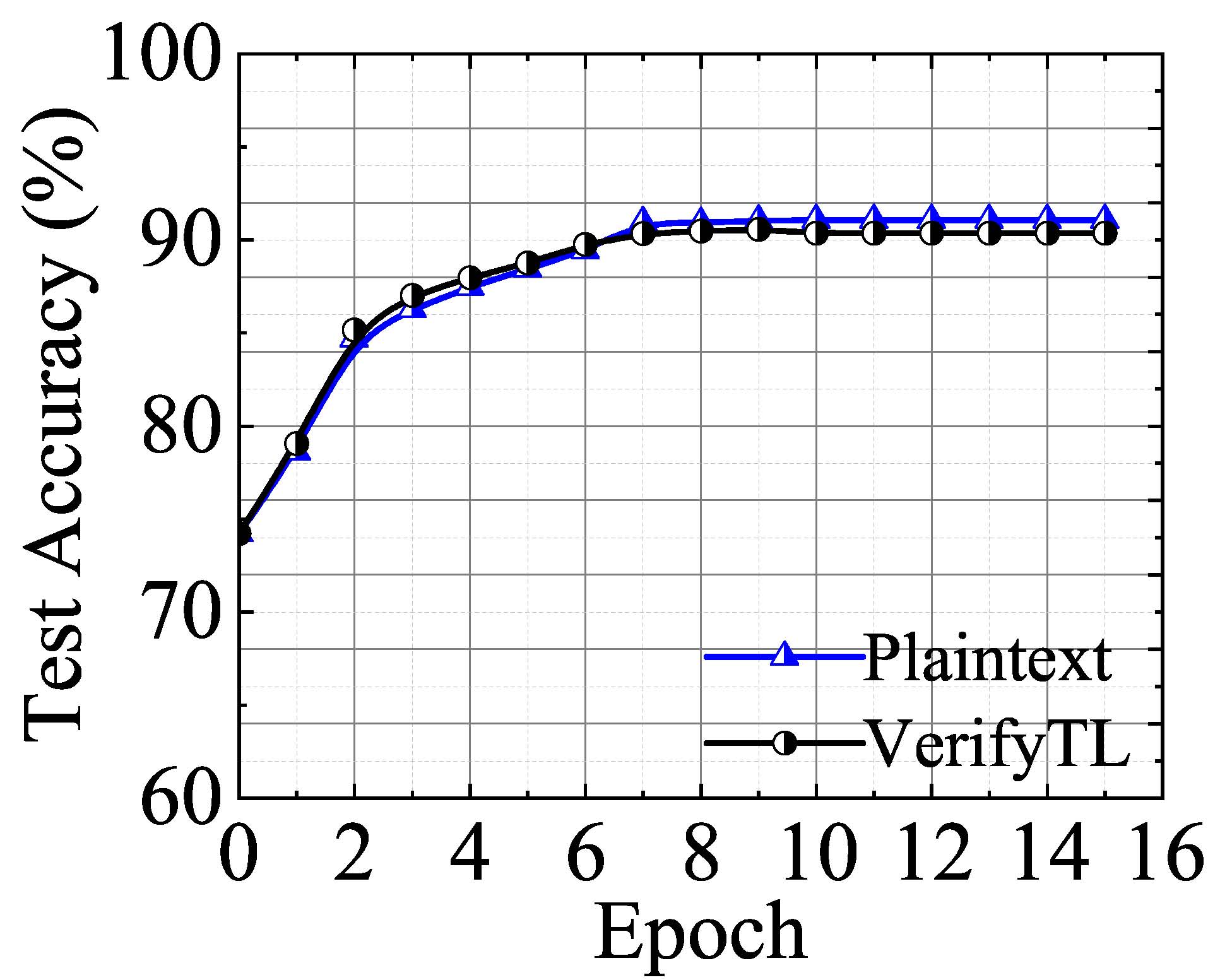}}
\subfigure[]{\includegraphics[width=1.4in]{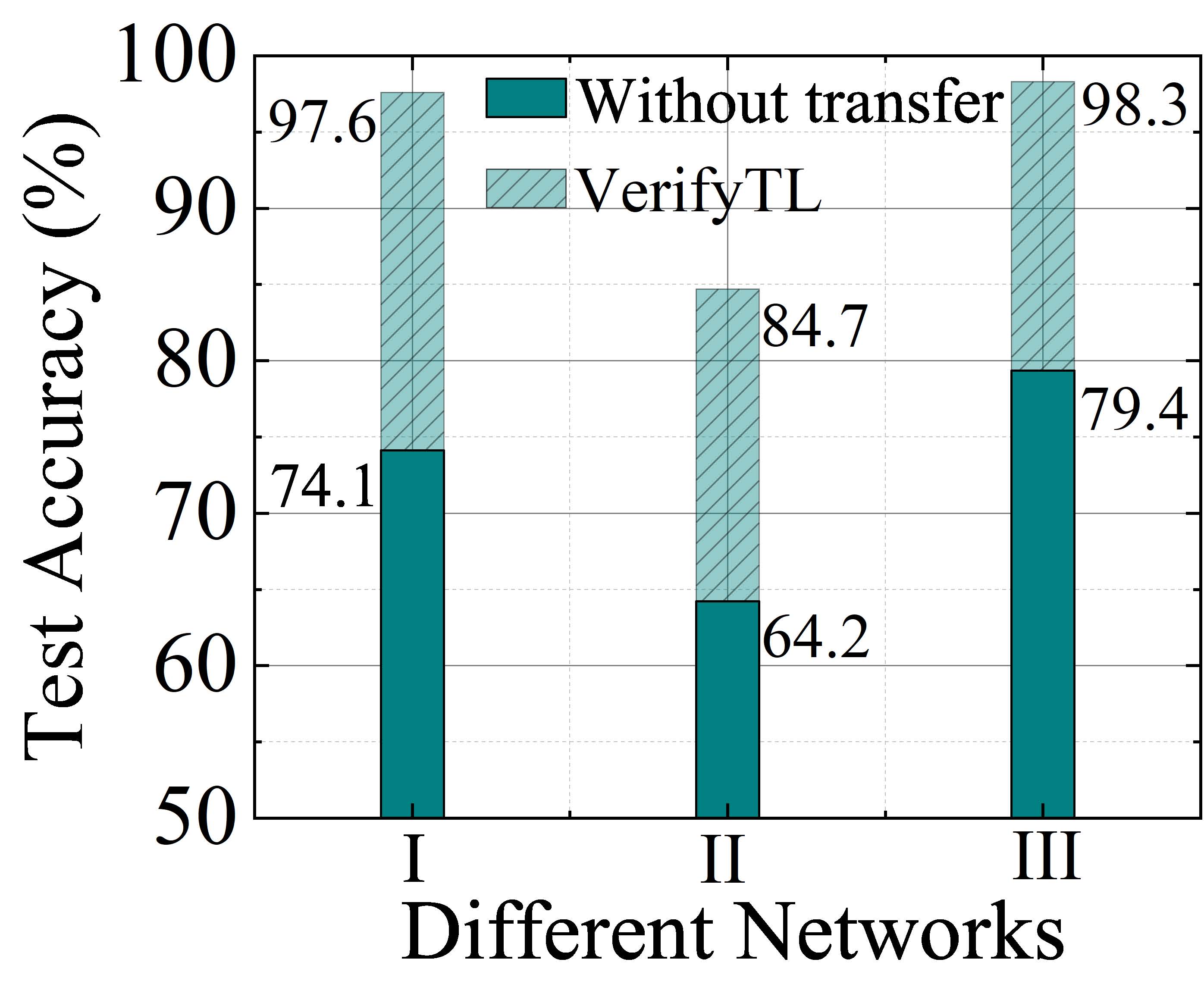}}

  \caption{Accuracy of {VerifyTL} on MNIST: (a) Accuracy of $\mathcal{D}_1$, with different degree vector $\Theta_1$. (b) Accuracy for a cross unit after $Pool_1$, $Pool_2$ and $Pool_1\&Pool_2$. (c) Accuracy of VerifyTL for different size of data domains. (d) Accuracy for VerifyTL and plaintext . (e) Accuracy for different networks.} \label{Effectivenessperformance}
\end{figure*}
\textbf{Network}. We adopt LeNet~\cite{lecun1998gradient} as our CNN architecture. The CNN model is denoted as \textsf{Network \textrm{I}}, which consists of $L=7$ layers such as 1 input layer, 2 convolution layers $Conv$, 2 pooling layers $Pool$, 1 full connection layer $Full$ and 1 output layer. The details are shown in Table~\ref{CNN}.

\begin{table}[!ht]
\centering \caption{\textsf{Network \textrm{I}} architecture}
\label{CNN}
\tabcolsep 6pt
\begin{threeparttable}[b] %
\begin{tabular*}{3.5in}{l||c|c|c|c}
\toprule
\hline
Layer & Parameters & Connections & Output& Unit \\ \hline
$Conv_1$&$156$ &$89,856$& $24*24*6$ &$5*5*6$ \\ \hline
$Pool_2$&$12$ &$4,320$& $12*12*6$ &$2*2*1$ \\ \hline
$Conv_3$ &$1,516$ &$97,024$& $8*8*12$ &$5*5*12$ \\ \hline
$Pool_4$&$32$ &$960$& $4*4*12$ &$2*2$ \\ \hline
$Full_5$ &$1,930$&$1,930$&$10*1*1$&12 \\ \hline
\bottomrule
\end{tabular*}
\begin{tablenotes}
\item \textbf{Notes}. The size of the output and unit on each layer is denoted as $h_l*w_l*c_l$. The unit is a convolution kernel on a $Conv$ layer, it is a pooling unit on a $Pool$ layer, and it is neurons on a $Full$ layer. $Full_5$ layer contains 1,930 trainable parameters and 10 neurons from the design of the output layer.
\end{tablenotes}
\end{threeparttable} %
\end{table}

\textbf{Parameters}. We set up the parameters in VerifyTL with a security parameter $\kappa=128$, precision $p=2^8$ and the size of data domains $n$ varies in the range $[2,10]$. All domains adopt the same CNN model $\mathcal{N}et$, where $batch$ $size=128$, $learning$ $rate$ = 0.01, $dropout=0.8$.

\subsection{Effectiveness}
Fig.~\ref{Effectivenessperformance} evaluates the test accuracy of VerifyTL according to the following factors. We adopt the 10-fold cross validation technique for CNN accuracy.

\textsf{Degree Vector}. Fig.~\ref{Effectivenessperformance}(a) depicts the test accuracy of a data domain $\mathcal{D}_1$ by varying with the value of degree vector $\Theta_1$, where the size of training data of $\mathcal{D}_1$ is 1K, and the size of training data of $\mathcal{D}_2$ is 5K, and VerifyTL runs over two data domains $\mathcal{D}_1$, $\mathcal{D}_2$. We discover that the accuracy increases with a bigger value of $\theta_{12}$. The reason is that a bigger $\theta_{12}$ means more knowledge can be employed from $\mathcal{D}_2$ to $\mathcal{D}_1$, and $\mathcal{D}_2$ owns larger training data for the accuracy improvement on $\mathcal{D}_1$.
The training accuracy is stable when $Epoch=10$, the accuracy is $97.6\%$ of $\Theta_1=(\theta_{11}=0.9,\theta_{12}=0.1)$, the accuracy is $93.4\%$ of $\Theta_1=(\theta_{11}=0.5,\theta_{12}=0.5)$, the accuracy is $91.3\%$ of $\Theta_1=(\theta_{11}=0.1,\theta_{12}=0.9)$, respectively.

\textsf{Transfer Unit}. Fig.~\ref{Effectivenessperformance}(b) describes the variation of the test accuracy with a cross unit $\Pi_{\text{CrossUnit}}^{\text{SPDZ}}$, which is adopted after $Pool_2$, $Pool_4$, and $Pool_2\&Pool_4$. The size of training samples on each domain is 1K. We discover that $\Pi_{\text{CrossUnit}}^{\text{SPDZ}}$ adopted after $Pool_2$ has almost the same the accuracy as $\Pi_{\text{CrossUnit}}^{\text{SPDZ}}$ adopted after $Pool_4$, while $\Pi_{\text{CrossUnit}}^{\text{SPDZ}}$ after $Pool_2\&Pool_4$ has better accuracy than others. Since there are more data representations are involved in the CNN building on two data domains,  $\Pi_{\text{CrossUnit}}^{\text{SPDZ}}$ is adopted after all pooling layers $Pool_2\&Pool_4$ to guarantee accuracy.

Based on above evaluations, in our default evaluation, we consider that the number of data domains $n\in[2,10]$, and 1K training samples on each data domain $\mathcal{D}_i$. We set the elements in a degree vector $\Theta_i$, where $\theta_t=0.1, \theta_s=1-\Sigma \theta_{t_{ij}}$. A transfer unit is adopted after each pooling layer (\ie $Pool_2\&Pool_4$). Besides, without emphasis, we employ the \textsf{Network \textrm{I}} architecture in VerifyTL.

\textsf{Data Domains}. Fig.~\ref{Effectivenessperformance}(c) shows that the test accuracy varies with the increasing size of data domains $n\in[2,5]$. We observe that the test accuracy increases with the growth of $n$. When $n=5$, the accuracy is $98.2\%$. This is because that more data representations are transferred among data domains with the increase of $n$. Thus, more knowledge can be adopted to improve accuracy on each data domain.

\textsf{Plaintext Comparison}. We show the accuracy comparison between VerifyTL and the proposed scheme over plaintexts (Fig.~\ref{Effectivenessperformance}(d)). We note that the accuracy of VerifyTL is similar to that of plaintexts with negligible accuracy difference. When $n=2$, the accuracy of VerifyTL is $90.4\%$, while the accuracy over plaintexts is $90.6\%$. This is because VerifyTL enables the privacy and verification over secret shares by adopting the approximation method to convert a rational number to the integer field, which may incur computation errors.

\textsf{Network architecture}. We evaluate the impact of different network architectures on the accuracy of VerifyTL. Fig.~\ref{Effectivenessperformance}(e) illustrates the test accuracy for different network architectures, where $Epoch=10$ and $n=5$.
We tested our evaluation over three kinds of network architectures (\textsf{Network \textrm{I, II, III}}), where \textsf{Network \textrm{II}} has the simplest architecture, while \textsf{Network \textrm{III}} has the most sophisticated one. The details are represented in Tables~\ref{CNN} -- \ref{CNN3}.
We notice the accuracies on all \textsf{Network \textrm{I, II, III}} of VerifyTL have significant improvement than those without transfer units. It is consistent with our scheme that VerifyTL is applicable to different kinds of CNN architectures.
\begin{table}[!ht]
\centering \caption{\textsf{Network \textrm{II}} architecture}
\label{CNN2}
\tabcolsep 6pt
\begin{threeparttable}[b] %
\begin{tabular*}{3.5in}{l||c|c|c|c}
\toprule
\hline
Layer & Parameters & Connections & Output& Unit \\ \hline
$Conv_1$&$520$ &$299,520$& $24*24*20$ &$5*5*20$ \\ \hline
$Pool_2$&$40$ &$14,400$& $12*12*20$ &$2*2*1$ \\ \hline
$Full_3$ &$288,000$&$288,000$&$10*1*1$&100 \\ \hline
\bottomrule
\end{tabular*}
\end{threeparttable} %
\end{table}
\begin{table}[!ht]
\centering \caption{\textsf{Network \textrm{III}} architecture}
\label{CNN3}
\tabcolsep 6pt
\begin{threeparttable}[b] %
\begin{tabular*}{3.5in}{l||c|c|c|c}
\toprule
\hline
Layer & Parameters & Connections & Output& Unit \\ \hline
$Conv_1$&$156$ &$122,304$& $28*28*6$ &$5*5*6$ \\ \hline
$Pool_2$&$12$ &$5,880$& $14*14*6$ &$2*2*1$ \\ \hline
$Conv_3$ &$1,516$ &$151,600$& $10*10*16$ &$5*5*16$ \\ \hline
$Pool_4$&$32$ &$2,000$& $5*5*16$ &$2*2*1$ \\ \hline
$Conv_5$ &$48,120$ &$48,120$& $120*1*1$ &$5*5*120$ \\ \hline
$Full_6$ &$10,164$&$10,164$&$10*1*1$&84 \\ \hline
\bottomrule
\end{tabular*}
\end{threeparttable} %
\end{table}

\subsection{Efficiency}
\begin{figure*}[!ht]
	  \centering

\subfigure[]{\includegraphics[width=1.73in]{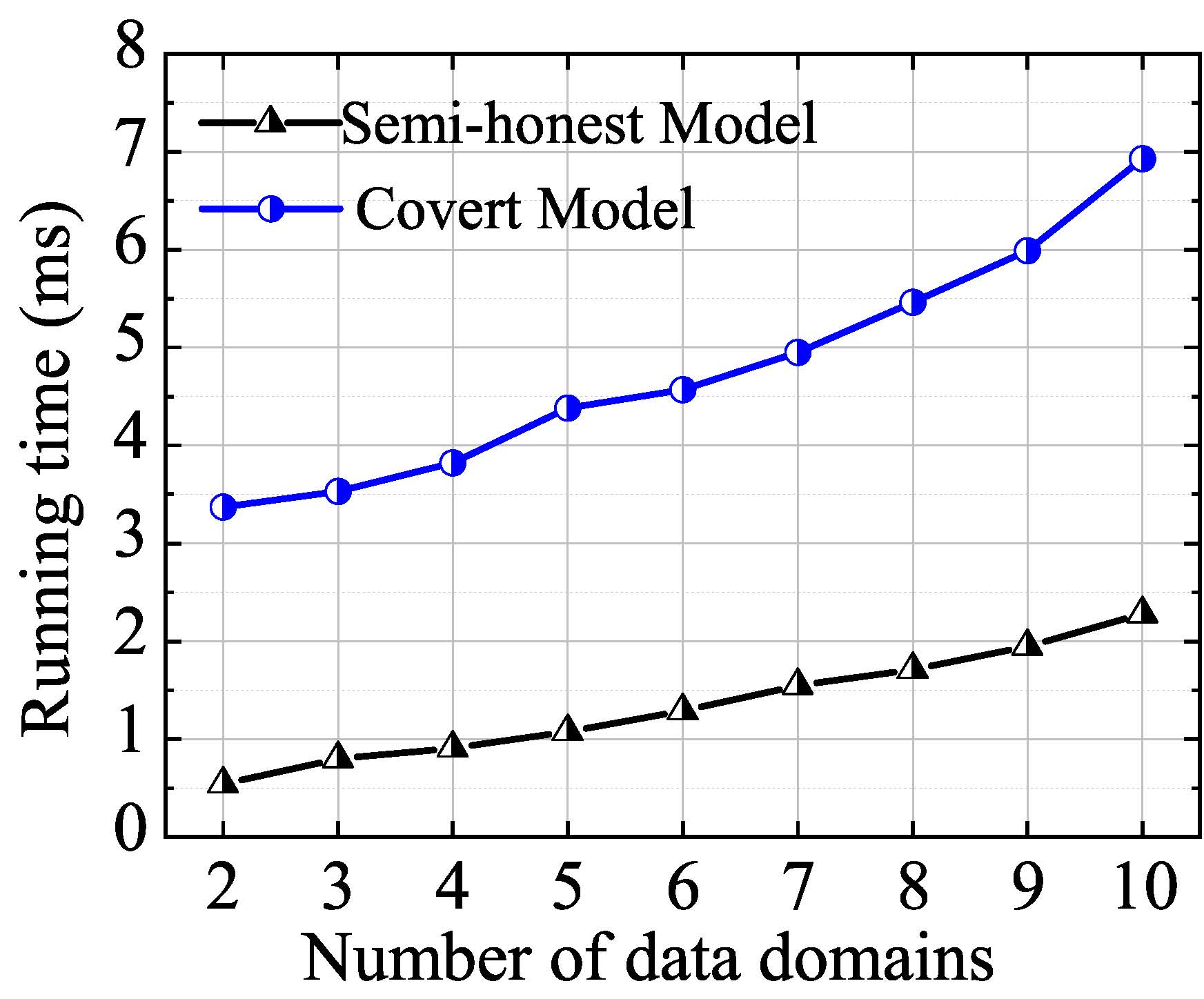}}
\subfigure[]{\includegraphics[width=1.6in]{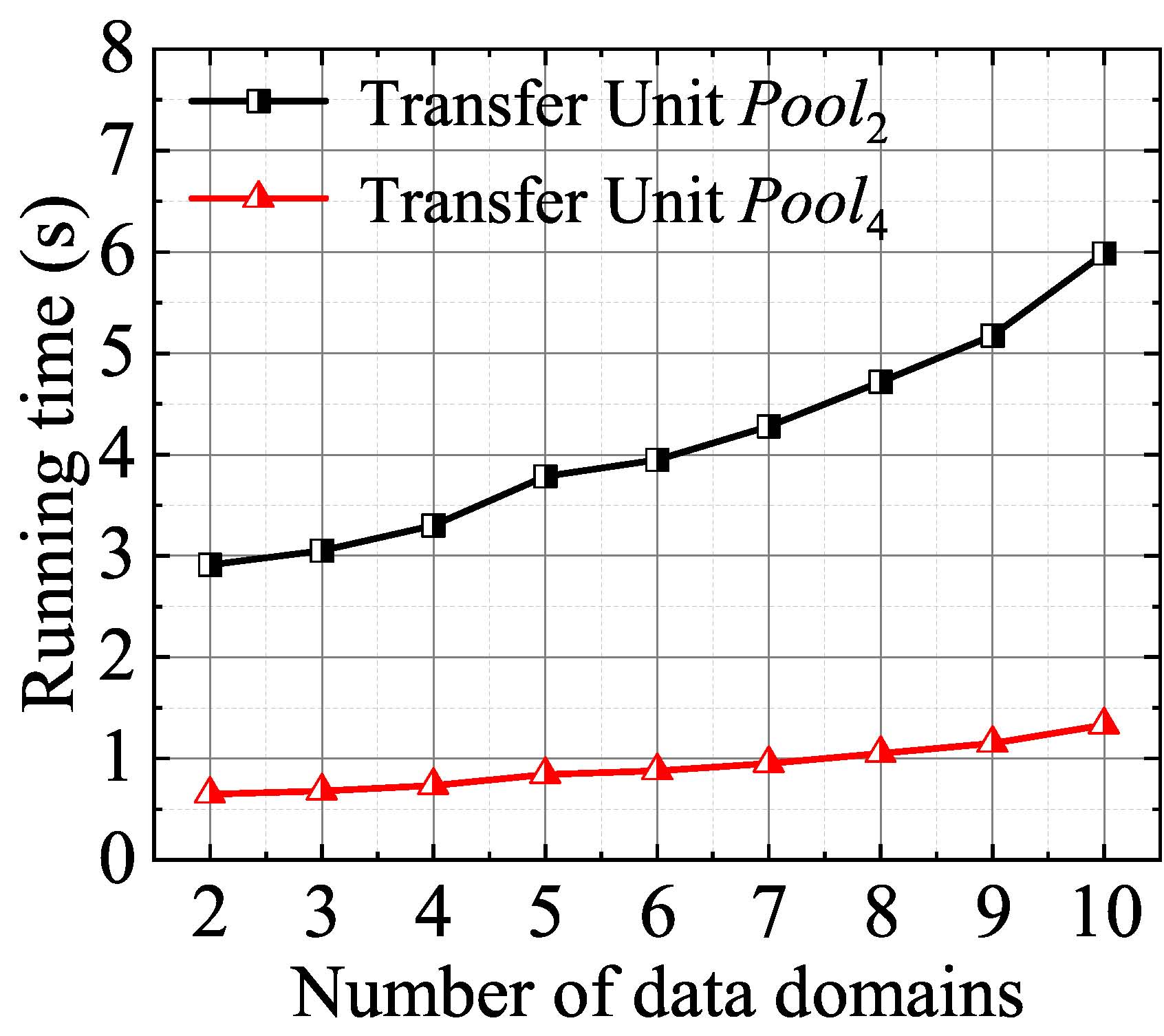}}
\subfigure[]{\includegraphics[width=1.75in]{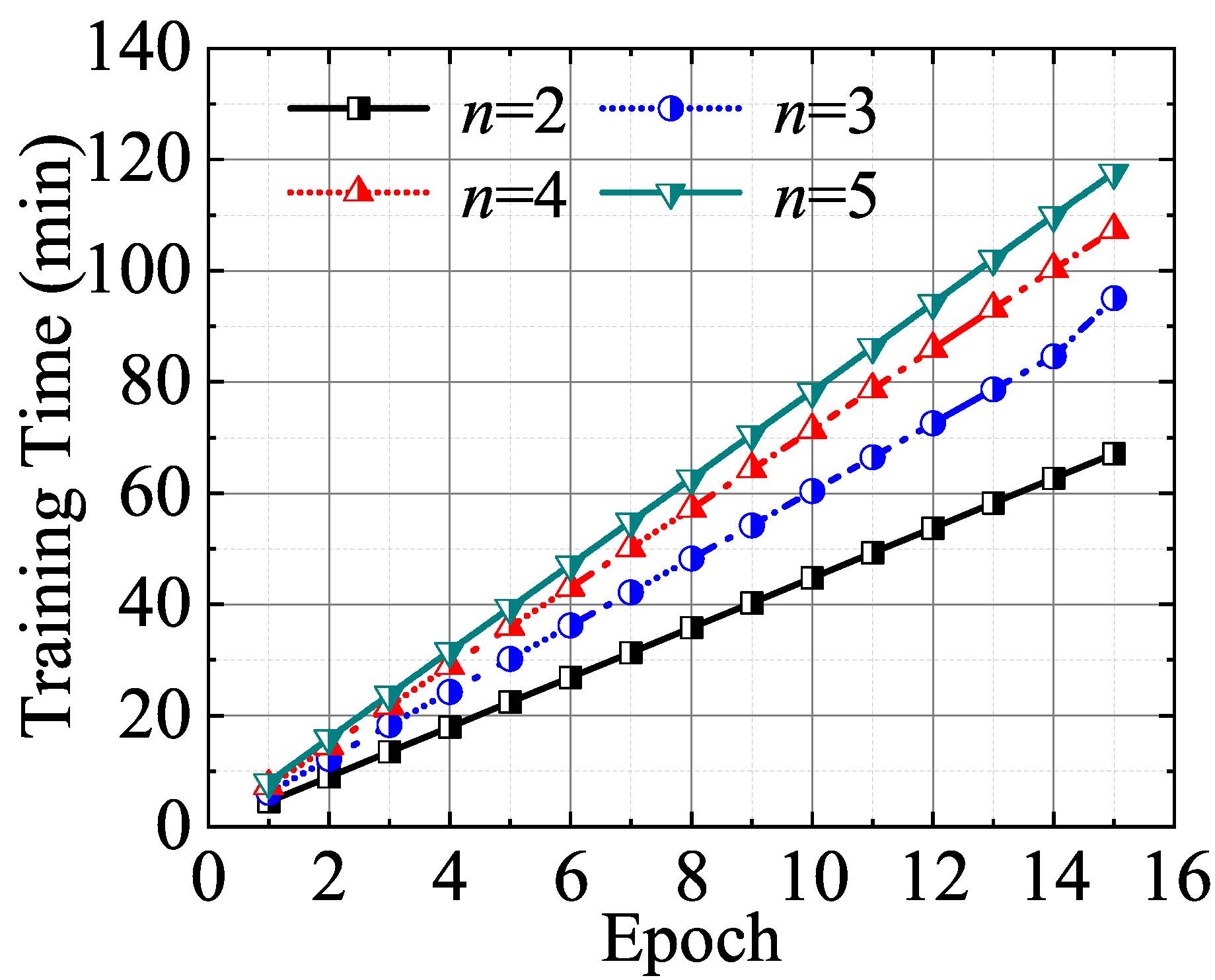}}
\subfigure[]{\includegraphics[width=1.67in]{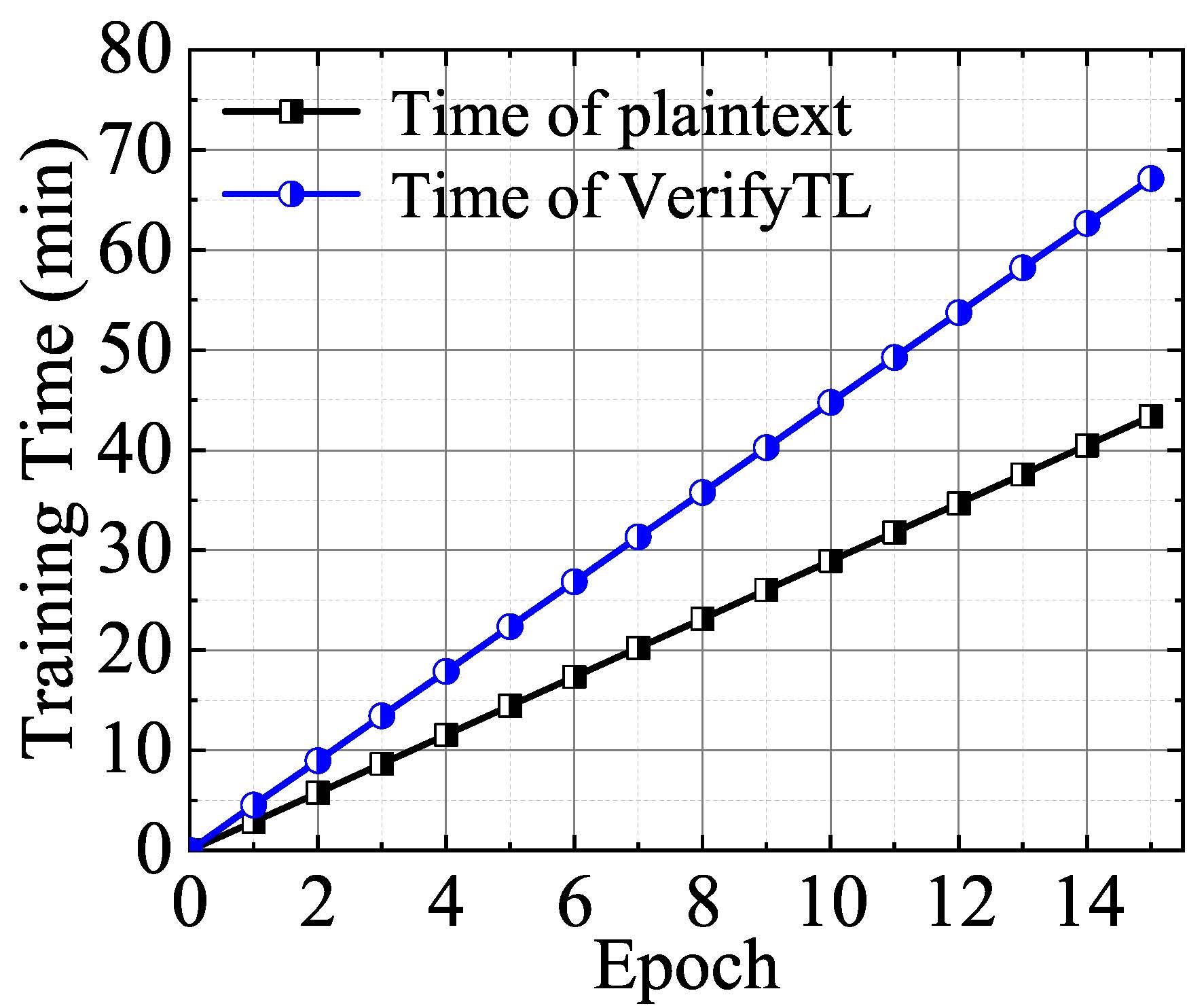}}
  \caption{Performance of {VerifyTL}: (a) Running time of \textsf{VectorMul}. (b) Running time of transfer unit. (c) Training time of VerifyTL with number of data domains (d) Training time of VerifyTL and plaintext ($n=2$). } \label{Efficiencyperformance}
\end{figure*}

\paratitle{Theoretical Analysis}. We analyze the computational complexity and communication complexity of a secure and verifiable transfer unit.

\textit{Computational complexity}. In a secure and verifiable transfer unit, the computational complexity mainly relies on $\textsf{VectorMul}$ as the setup stage is operated offline, thus computation and communication overheads are ignored.
Since the computational complexity of $\textsf{VectorMul}$ depends on the size of vectors $\mathcal{V}_i$ and $\Theta_i$, $n$ times SPDZ multiplication operations are involved in a $\textsf{VectorMul}$, which costs $\mathcal{O}(n)$ in a SPDZ multiplication with linear opearations. Thus, $\textsf{VectorMul}$ costs $\mathcal{O}(n^2)$ time over all domains.
As the size of transferred activation maps is $\mathcal{X}^l\leftarrow \mathbb{R}^{h_l\times w_l\times c_l}$, the size of transferred degree vector $\Theta_i$ is $n$, it involves ${h_l\times w_l\times c_l}$ times $\textsf{VectorMul}$, thus a secure and verifiable transfer unit costs $\mathcal{O}(n^2*h_l*w_l*c_l)$ time over all domains.

\textit{Communication complexity}. A secure and verifiable collaborative transfer unit has communication complexity $\mathcal{O}(n^2(h_l*w_l*c_l+n))$:
Since the communication complexity relies on the size of activation maps and the number of data domains, each data domains communicates $\mathcal{O}(n(h_l*w_l*c_l+n))$ data items, where the size of a transferred activation map is $h_l*w_l*c_l$, the size of a transferred degree vector is $n$.

\paratitle{Experimental Analysis}.
Fig.~\ref{Efficiencyperformance} demonstrates the execution time of the following sections of VerifyTL, which is an average over $100$ trials.

{\textsf{VectorMul}}. Fig.~\ref{Efficiencyperformance}(a) shows the running time of \textsf{VectorMul} with the semi-honest model and covert security model, respectively. We observe that the running time increases with the growth of data domains $\mathcal{D}$. Since a bigger size of vector $\mathcal{X}^l$ and $\Theta_i$ is involved, more times of SPDZ multiplication are executed. Also, \textsf{VectorMul} costs more overhead in our covert security model to guarantee the verification of computation results compared with \textsf{VectorMul} without \textsf{MACCheck} in the semi-honest setting. This is expected, as verification process is required to spend execution time.
It creates a tradeoff between security and efficiency as a covert model provides a higher security level than a semi-honest. When $n=10$, \textsf{VectorMul} in our threat model costs 6.93 ms, while \textsf{VectorMul} in semi-honest model costs 2.14 ms. Thus, the increase of verification time is within an acceptable limit for implementing the covert security model. 

\textsf{Transfer Unit}. In Fig.~\ref{Efficiencyperformance}(b), we discover that the execution time of a transfer unit is affected by the size of data domains and the size of inputs. The running time of a transfer unit after $Pool_2$ is more than that of it after $Pool_4$. The reason is that the input of a transfer unit after a $Pool_2$ layer is $\mathcal{X}^{Pool_2}$ with the size of $12*12*6$, while the input of a transfer unit after a $Pool_4$ layer is $\mathcal{X}^{Pool_4}$ with the size of $4*4*12$, thus more elements are involved in a transfer unit after $Pool_2$ for SPDZ computation. When $n=10$, a transfer unit after $Pool_2$ costs $5.98$ s, while a transfer unit after $Pool_4$ costs $1.33$ s.

\textsf{VerifyTL}. Fig.~\ref{Efficiencyperformance}(c) depicts the influence of the size of data domains on the training time. We notice that the training time of VerifyTL is increased as the growth of data domains.
It represents that more activation maps are transferred to tune more local CNN model with the increase of $\mathcal{D}$. When $Epoch=15$, the training time is 117.7 min with $n=5$, while the training time is 67.1 min with $n=2$.

In Fig.~\ref{Efficiencyperformance}(d), we notice that the running time of VerifyTL is larger than that of the proposed scheme over plaintexts, where $n=2$. When $Epoch=10$, the training time is 44.8 min and that of plaintexts is 19.2 min. This is because VerifyTL implements a transfer unit over secret shares with SPDZ computation to guarantee privacy and verification.
\subsection{Comparative Evaluations}
We compare VerifyTL with original learning without transfer units~\cite{lecun1998gradient}, federated learning~\cite{sattler2019robust} and cross-stitch transfer learning~\cite{misra2016cross}, where federated learning is a kind of distributed machine learning scheme.

\begin{table}[!ht]
\centering \caption{Test accuracy and training time comparison}
\label{AccComparison}
\tabcolsep 3pt
\begin{threeparttable}[b] %
\begin{tabular*}{3.5in}{l||c|c|c}
\toprule
  \hline
   {Method}&Accuracy &Training time&Threat model \\ \hline
  Original learning  & 74.6\% ({\color {red} 73.6\%})  &  \textbf{0.32 h} &---\\ \hline
  Cross-stitch~\cite{misra2016cross}& 87.4\% ({\color {red} 86.8\%})  &   0.37 h &---\\ \hline
  Federated learning~\cite{sattler2019robust}& 92.3\% ({\color {red}90.2\%})  &   14.8 h & Semi-honest\\ \hline
  VerifyTL& \textbf{98.2\% ({\color {red} 97.6\%})} & 1.31 h&\textbf{Covert} \\ \hline

  \bottomrule
\end{tabular*}
\end{threeparttable}
\begin{tablenotes}
\item \textbf{Notes}. Black text is test accuracy on MNIST, {\color {red}red text} is test accuracy on Fashion MNIST, where the size of training data on each domain is 1K samples, and $Epoch=10$. Original learning is a \textsf{Network \textrm{I}} model trained solely on a single data domain without any collaborative transfer units, where $n=1$. In~\cite{misra2016cross}, $n=2$. In~\cite{sattler2019robust} and VerifyTL, $n=5$.
\end{tablenotes}
\end{table}

Based on Table~\ref{AccComparison}, we conclude that VerifyTL provides a stronger security model and achieves outstanding accuracy results that can rival with other approaches.
The accuracy of the original learning with a \textsf{Network \textrm{I}} model without any collaborative transfer units is compared with VerifyTL. VerifyTL implements a significant accuracy improvement and provides privacy and verification with an acceptable training time.
Also, compared with~\cite{misra2016cross}, VerifyTL maintains outstanding accuracy and extends cross transfer learning from the two-domain setting to the multi-domain setting with strong privacy preservations, where \cite{misra2016cross} runs over plaintexts without privacy preservations.
Besides, VerifyTL performs better than federated learning~\cite{sattler2019robust} in both security, efficiency and effectiveness.
In federated learning~\cite{sattler2019robust}, a data domain is required to securely outsource trainable parameters at each layer to a semi-honest central server. There are total 3,646 parameters in a \textsf{Network \textrm{I}} model, which costs huge computation overhead for secure outsourcing. The reason for the computation overhead is that \cite{sattler2019robust} is based on Paillier cryptosystem, which involves more expensive exponent arithmetic to guarantee the privacy by encrypting the large size of transmitted CNN parameters during each training epoch. Unfortunately, federated learning cannot support covert security, of which the correctness of behaviours among distributed data domains and the central server cannot be guaranteed. Once a covert adversary corrupts $n-1$ data domains, it will lead to incorrect training to undermine the accuracy.

\section{Conclusion}
In this paper, we proposed a secure and verifiable collaborative transfer learning (VerifyTL) scheme. The scheme facilitates the collaborative transfer over extracted knowledge among multiple data domains in a strong privacy preserving manner, and allows verification against dishonest majority for implementing covert security. We mathematically proved the security of VerifyTL, as well as evaluating its performance using two real-world datasets MNIST and Fashion MNIST, \ie the performance gains with VerifyTL up to + 23.6\% for MNIST and +24.0\% for Fashion MNIST compared with original learning.

\section*{Acknowledgments}
This work was supported by the Key Program of NSFC (No. U1405255), the Shaanxi Science \& Technology Coordination \& Innovation Project (No. 2016TZC-G-6-3), the National Natural Science Foundation of China (No. 61702404, No. 61702105, No. U1804263), the China Postdoctoral Science Foundation Funded Project (No. 2017M613080), the Fundamental Research Funds for the Central Universities (No. JB171504, No. JB191506), the National Natural Science Foundation of Shaanxi Province (No. 2019JQ-005).
\bibliographystyle{IEEEtran}
\bibliography{bibfile}

\begin{IEEEbiography}[{\includegraphics[width=1in,height=1.25in,clip,keepaspectratio]{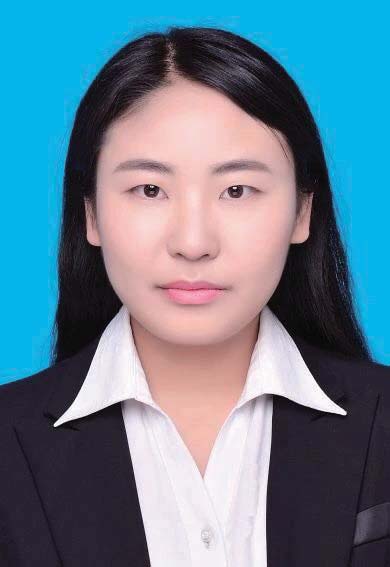}}]{Zhuoran Ma}
received the B.E. degree from the School of Software Engineering, Xidian University, Xi¡¯an, China, in 2017. She is currently a Ph.D candidate with the Department of Cyber Engineering, Xidian University. Her current research interests include data security and secure computation outsourcing.
\end{IEEEbiography}
\vspace{-11 mm}
\begin{IEEEbiography}[{\includegraphics[width=0.9in,height=1.25in,clip,keepaspectratio]{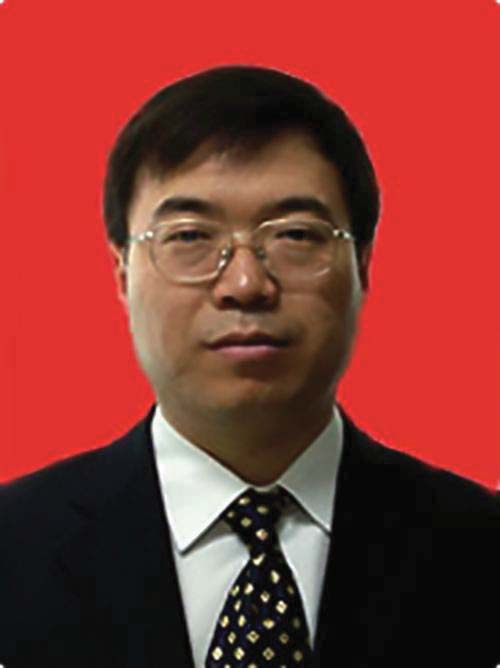}}]{Jianfeng Ma}
 received the Ph.D. degree in computer software and telecommunication engineering from Xidian University, Xi'an, China, in 1988 and 1995, respectively. From 1999 to 2001, he was a Research Fellow with Nanyang Technological University of Singapore. He is currently a professor and a Ph.D. Supervisor with the Department of Computer Science and Technology, Xidian University, Xi'an, Chian. He is also the Director of the Shaanxi Key Laboratory of Network and System Security. His current research interests include information and network security, wireless and mobile computing systems, and computer networks.
\end{IEEEbiography}
\vspace{-10 mm}

\begin{IEEEbiography}[{\includegraphics[width=0.9in,height=1.25in,clip,keepaspectratio]{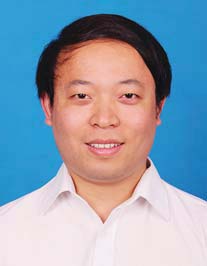}}]{Yinbin Miao}
 received the B.E. degree with the Department of Telecommunication Engineering from Jilin University, Changchun, China, in 2011, and Ph.D. degree with the Department of Telecommunication Engineering from xidian university, Xi'an, China, in 2016. He is currently a Lecturer with the Department of Cyber Engineering in Xidian university, Xi'an, China. His research interests include information security and applied cryptography.
\end{IEEEbiography}

\begin{IEEEbiography}[{\includegraphics[width=0.9in,height=1.25in,clip,keepaspectratio]{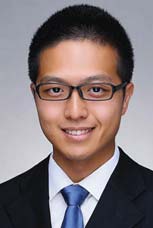}}]{Ximeng Liu} (S¡¯13-M¡¯16) received the B.Sc. degree in electronic engineering from Xidian University, Xi¡¯an, China, in 2010 and the Ph.D. degree in Cryptography from Xidian University, China, in 2015. Now he is the full professor in the College of Mathematics and Computer Science, Fuzhou University. Also, he was a research fellow at the School of Information System, Singapore Management University, Singapore. He has published more than 200 papers on the topics of cloud security and big data security including papers in IEEE TOC, IEEE TII, IEEE TDSC, IEEE TSC, IEEE IoT Journal, and so on. He awards ¡°Minjiang Scholars¡± Distinguished Professor, ¡°Qishan Scholars¡± in Fuzhou University, and ACM SIGSAC China Rising Star Award (2018). His research interests include cloud security, applied cryptography and big data security. He is a member of the IEEE, ACM, CCF.
\end{IEEEbiography}

\begin{IEEEbiography}[{\includegraphics[width=0.9in,height=1.25in,clip,keepaspectratio]{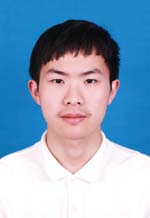}}]{Wei Zheng} will receive the B.E. degree with the Department of Information Security from Nanchang University, Nanchang, China, in 2020. He will pursue his M.E degree with the Department of Cyber Engineering from Xidian University, Xi¡¯an, China. His research interests include information security
and applied cryptography.
\end{IEEEbiography}

\begin{IEEEbiography}[{\includegraphics[width=1in,height=1.25in,clip,keepaspectratio]{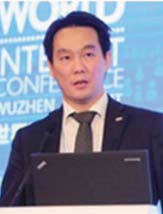}}]{Kim-Kwang Raymond Choo}
(SM'15) received the Ph.D. in Information Security in 2006 from Queensland University of Technology, Australia. He currently holds the Cloud Technology Endowed Professorship at The University of Texas at San Antonio (UTSA). He is the recipient of various awards including the UTSA College of Business Col. Jean Piccione and Lt. Col. Philip Piccione Endowed Research Award for Tenured Faculty in 2018, ESORICS 2015 Best Paper Award. He is an Australian Computer Society Fellow, and an IEEE Senior Member.
\end{IEEEbiography}

\begin{IEEEbiography}[{\includegraphics[width=1in,height=1.25in,clip,keepaspectratio]{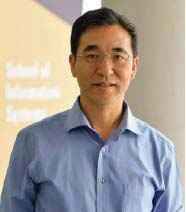}}]{Robert H. Deng}
(F'16) has been a Professor with the School of Information Systems, Singapore Management University, since 2004. His research interests include data security and privacy, multimedia  security, and network and system security. He was  an Associate Editor of the IEEE TRANSACTIONS ON {INFORMATION FORENSICS AND SECURITY} from  2009 to 2012. He is currently an Associate Editor of the IEEE TRANSACTIONS ON DEPENDABLE AND  SECURE COMPUTING and Security and Communication Networks (John Wiley). He is the cochair of the Steering Committee of the ACM Symposium on Information, Computer and Communications Security. He received the University Outstanding Researcher Award from the National University of Singapore in 1999 and the Lee Kuan Yew Fellow for Research Excellence from Singapore  Management University in 2006. He was named Community Service Star and Showcased Senior Information Security Professional by (ISC)$^2$ under its Asia-Pacific Information Security Leadership Achievements Program in 2010.
\end{IEEEbiography}

\end{document}